\documentclass[a4paper,11pt]{article}
\usepackage{mathpazo}
\usepackage{nicefrac}
\usepackage{rotating}
\usepackage[margin=1.2in]{geometry}
\usepackage{algorithm}
\usepackage{algcompatible}
\usepackage{graphics}
\usepackage{setspace}
\usepackage{hyperref}
\usepackage{relsize}
\usepackage{hyperref}
\hypersetup{colorlinks=true,linkcolor=black,citecolor=[rgb]{0.5,0,0}}
\usepackage{color}
\usepackage{amsthm}
\usepackage{graphicx}
\usepackage{amsmath}
\usepackage{enumerate}
\usepackage{amssymb}
\usepackage{diagbox}
\usepackage{caption}
\usepackage{bbm}

\usepackage{subcaption}
\usepackage{tikz}
\usepackage{thmtools}
\usepackage{thm-restate}
\usepackage{comment}
\usepackage{cleveref}

\tikzstyle{vertex}=[circle, draw, inner sep=0pt, minimum size=6pt]

\usetikzlibrary{patterns,positioning,decorations.pathmorphing,decorations.pathreplacing}
\usetikzlibrary{arrows,automata,snakes}
\usepackage[classfont=serif,
langfont=roman,
funcfont=italic]{complexity}

\usepackage{tcolorbox}
\tcbuselibrary{breakable}

\newtheorem{theorem}{Theorem}[section]
\newtheorem{utopic}{Utopic Theorem}

\newtheorem{corollary}[theorem]{Corollary}
\newtheorem{lemma}[theorem]{Lemma}
\newtheorem*{lemma*}{Lemma}

\newtheorem{claim}[theorem]{Claim}

\newtheorem{definition}[theorem]{Definition}
\newtheorem{remark}[theorem]{Remark}

\usepackage{mleftright}
\usepackage{epigraph}

\usepackage{xspace}

\renewcommand{\tilde}{\widetilde}
\newcommand{\zo}{\{0,1\}}

\renewcommand{\S}{\mathcal{S}}

\newcommand{\Fac}{\mathsf{Factor}\xspace}
\newcommand{\EOL}{\mathsf{EoL}\xspace}
\newcommand{\LCS}{$\mathsf{LCS}$\xspace}
\newcommand{\Edit}{$\mathsf{Edit}$-$\mathsf{Distance}$\xspace}
\newcommand{\MAXSAT}{$\mathsf{MaxSAT}$\xspace}
\newcommand{\MAXCLIQUE}{$\mathsf{MaxCLIQUE}$\xspace}

\newcommand{\LCSsig}{\mathsf{LCS}_\Sigma}
\newcommand{\LCSxi}{\mathsf{LCS}_\Xi}
\newcommand{\Mult}{\mathsf{Mult}}

\newcommand{\disp}{\text{disp}}

\renewcommand{\p}{\mathsf{fail}}

\renewcommand{\D}{\mathcal{D}}
\newcommand{\N}{\mathbb{N}}

\newcommand{\card}[1]{\left\lvert #1 \right\rvert}
\renewcommand{\E}{\mathbb{E}}
\renewcommand{\R}{\mathbb{R}}

\newcommand{\ED}{$\mathsf{ED}$\xspace}		
\newcommand{\KS}{\mathsf{KS}\xspace}

\newcommand{\CNF}{$\mathsf{CNF}$\xspace}		
\newcommand{\ETH}{$\mathsf{ETH}$\xspace}		
		
\newcommand{\EDsig}{\mathsf{ED}_\Sigma}		
\newcommand{\EDxi}{\mathsf{ED}_\Xi}

\renewcommand{\G}{\mathsf{Gen}}
\newcommand{\Dec}{\mathsf{Dec}}
\newcommand{\SOL}{\mathsf{Sol}}
\newcommand{\m}{\Delta}
\newcommand{\goal}{\mathsf{goal}}

\newcommand{\sett}[2]{\{#1 | #2\}}
\newcommand{\bxi}{\mathbf{z}}

\newcommand{\eqdef}{:=}
\newcommand{\set}[1]{\{#1\}}

\title{\textbf{Hardness Amplification of Optimization Problems}}
\date{}
\author{\iftrue Elazar Goldenberg\vspace{0.1cm}\\
 The Academic College of Tel Aviv-Yaffo\vspace{0.1cm} \\
\texttt{elazargo@mta.ac.il} \vspace{0.5cm}
\and 
Karthik C.\ S.\footnote{This work was partially supported by Irit Dinur's  ERC-CoG grant 772839.   } \vspace{0.1cm}\\
 Weizmann Institute of Science\vspace{0.1cm}\\
   \texttt{karthik.srikanta@weizmann.ac.il}
\fi}
\begin{document}
\maketitle  

\begin{abstract}
In this paper, we prove a general hardness amplification scheme for optimization problems based on the technique of direct products. \vspace{0.1cm}

We say that an optimization problem $\Pi$   is direct product feasible if it is possible to efficiently aggregate any $k$ instances of $\Pi$ and form one large instance of $\Pi$ such that given an optimal feasible solution to the larger instance, we can efficiently find optimal feasible solutions to all the $k$ smaller instances.  
Given a direct product feasible optimization problem $\Pi$, our hardness amplification theorem may be informally stated as follows: \vspace{-0.45cm}

\begin{center}\begin{tcolorbox}[breakable,colback=gray!10!white,colframe=white,width=0.75\textwidth]\centering
If there is a distribution $\mathcal{D}$ over instances of $\Pi$   of size $n$ such that\\\vspace{0.1cm} every randomized algorithm running in time $t(n)$ fails to solve\\\vspace{0.1cm} $\Pi$ on $\frac{1}{\alpha(n)}$ fraction of inputs sampled from $\mathcal{D}$,\\\vspace{0.1cm}
 then, assuming some relationships on $\alpha(n)$ and $t(n)$,\\\vspace{0.1cm}
there is  a distribution $\mathcal{D}'$ over instances of $\Pi$  of size $O(n\cdot \alpha(n))$ such that every randomized algorithm running in time $\frac{t(n)} {\poly(\alpha(n))}$ fails to solve $\Pi$  on $\nicefrac{99}{100}$ fraction of inputs sampled from $\mathcal{D}'$.\end{tcolorbox}
\end{center}\vspace{-0.2cm}

As a consequence of the above theorem, we show hardness amplification of problems in various classes such as \NP-hard problems like Max-Clique, Knapsack, and Max-SAT, problems in \P\ such as Longest Common Subsequence, Edit Distance, Matrix Multiplication, and even problems in \TFNP\ such as Factoring and computing Nash equilibrium. 
\end{abstract}

\clearpage

\section{Introduction}\label{sec:intro}

The widely believed conjecture $\P\neq \NP$ asserts that the class  $\NP$ cannot be decided efficiently on the worst-case. That is, no polynomial time algorithm can decide the satisfiability of a \CNF formula on \textit{every} instance. However, the worst case hardness of $\NP$ still does not clarify its average-case hardness: how hard is to decide the satisfiability on a uniformly random instance.

Studying the average-case hardness of $\NP$ has a two-fold motivation. First, it may provide a more meaningful explanation than worst-case complexity about the intractability of \NP-hard instances actually encountered in practice. In other words, if $\NP$ is hard only on the worst-case, then the theory of worst-case complexity that has been extensively developed over the last fifty years, might be not be a good reflection of reality. Second, hardness on average is the cornerstone of modern cryptography as the security of any nontrivial cryptosystem requires some computational problem to be average-case hard (for some nice distribution).  Additionally, showing average-case hardness for functions is a stepping stone towards proving strong derandomization results and the construction of pseudorandom generators.

The study of hardness amplification is the task of connecting the worst-case and average-case hardness. More specifically, based on a worst-case hardness (assumption) one would like to prove the average-case hardness of the problem.

\subsection{Utopic Theorem of Hardness Amplification}\label{sec:dream}

A utopic theorem in the context of hardness amplification would assert that if a function is hard in the worst-case then it implies the average-case hardness for the same function against algorithms with essentially the same running time complexity. More formally it would look as follows:

\begin{utopic}[A Utopic Hardness Amplification Theorem]
	Let $\set{f_n}_{n\in \N}$ be a family of functions. Assume that every algorithm 
	 running in time $t(n)$, fails to compute $f_n$ on at least $\gamma(n)$ fraction of inputs.
	Then there exists a family $\set{g_n}_{n\in \N}$  of functions, such that every algorithm running in time $t'(n)$, fails to compute $g_n$ on at least $\gamma'(n)$ fraction of inputs.
	
	Ideally, we would like to achieve the above amplification for the following parameters\footnote{In order to succinctly specify the desirable parameters of a hardness amplification theorem, we assume here that $f_n$ and $g_n$ are Boolean functions.}
\begin{enumerate}
\item $ \gamma(n)=O(\nicefrac{1}{2^n})$ and $\gamma'(n)=\nicefrac{1}{2}-O(\nicefrac{1}{2^n})$,
\item $t'(n)\approx t(n)$,
\item $\set{f_n}_{n\in \N}=\set{g_n}_{n\in \N}$.
\end{enumerate}	
\end{utopic}

We briefly elaborate here why we would like the above three setting of parameters in our utopic hardness amplification theorem. 
Item 1 would yield a worst-case to average-case reduction, and therefore extend all the lower bounds and hardness results that have been achieved in the theory of worst-case complexity for $f$ to the  average-case complexity of $g$. In fact, achieving $\gamma'(n)=\nicefrac{1}{2}-O(\nicefrac{1}{2^n})$ would imply that no algorithm running in time $t'(n)$ can do much better than randomly guessing the output.  
Item 2 would imply that our worst-case complexity lower bounds  meaningfully translate to lower bounds in the average-case.
Item 3 expresses the notion of self-reducability: if we are interested in understanding the average complexity of a problem, our hardness amplification theorem should enable us to do so by analyzing the worst-case complexity of the \emph{same} problem. In summary, obtaining a hardness amplification result satisfying the three items  is in a sense an attempt to bridge the gap between theory and practice. 
Finally, we remark that our utopic theorems would gain more importance if the family of functions for which we show hardness amplification are natural (in some broad sense). 

Specifically, if we prove such a theorem for the family of deciding satisfiability of \CNF formulas, then we get that the assumption that $\P \neq \NP$ implies that every polynomial time algorithm fails to decide satisfiability on  slightly more that half of the \CNF formulas -- a highly non-trivial and very desirable result that would pave the way for the construction of one-way functions from (weak) worst-case assumptions. However, as we wake up from the dream of a utopia, one may wonder if such a result can even be achieved \cite{BT06}.

Remarkably, nearly three decades ago, Lipton \cite{L89,CPS99}   proved the above type of (utopic) theorem for the function of computing the permanent (a \#\P-complete problem) against probabilistic polynomial time algorithms. Trevisan and Vadhan~\cite{TV07} following a line of works \cite{BFNW93,IW97,STV01} were almost  able to prove such an amplification result for the class \EXP\ (they couldn't achieve Item 3). 
For the class $\NP$ we are far from proving a strong hardness amplification result, and there are some known barriers while trying to convert worst-case $\NP$-hardness into average-case hardness (see e.g. \cite{BTsurvey}). More recently,  strong hardness amplification results have been proved for functions in \P\ \cite{BRSV17,GR17,GR18}. We also note that hardness amplification results have also been shown for one-way functions \cite{Y82,GILVZ90,BR13}.

Given the above state-of-the-art picture, we raise a few natural questions and address them in this paper.  There are many problems that are hard in the worst-case but easy on average. For example, 3-coloring is a well-known \NP-hard problem, but it is an easy exercise to show that it can
be solved in linear time with high probability on a random graph. This motivates us to distinguish within worst-case hard problems as to which of them remain hard on average. One way to go about this task is to identify which worst-case hard problem admits a hardness amplification theorem. 

\begin{center}
\textit{For which problems can we amplify hardness?\\ Can we identify a mathematical structure that allows us to amplify hardness?}
\end{center}

The latter question has been implicitly addressed  in literature (for example, if the problem has algebraic structure like in the case of computing permanent \cite{L89} or counting $k$-cliques \cite{GR18}), but are quite specific and not broad enough to capture the class of problems that we believe are hard on average. In Section~\ref{sec:resultmain} we address the above two questions. 

Next, we turn our attention to \NP-hard problems. In a beautiful paper, O'Donnell \cite{O04} initiated the study of (non-uniform) hardness amplification in \NP. His result was improved by \cite{HVV06} who showed that if for some function in \NP\  (with $n$ inputs) we have that any $s(n)$ size circuit fails to compute the function correctly on $\nicefrac{1}{\poly(n)}$ fraction of the inputs then, the hardness can be amplified to show that there is some function in $\NP$ such that any $s(\sqrt{n})^{\Omega(1)}$ size circuit fails to compute the function on $\nicefrac{1}{2}-\nicefrac{1}{s(\sqrt{n})^{\Omega(1)}}$ fraction of the inputs. However, the best uniform hardness amplification results (against algorithms as opposed to circuits) that have been achieved do not match the parameters of \cite{HVV06}: Trevisan \cite{Tre05,BKS06} improving on his previous work \cite{T03} showed that we can amplify hardness from $\nicefrac{1}{\poly(n)}$ to $\nicefrac{1}{2}-\nicefrac{1}{\polylog(n)}$ for \NP\ against randomized polynomial time algorithms (later extended to deterministic algorithms in \cite{GG11}). 
%\paragraph{\NP.} \cite{GST07,GT07}
However, it is important to note that all these hardness amplification results are for decision problems, and this leads us to our next question, do we gain anything by moving to search problems, or more precisely to the focus of this paper, to optimization problems?

\begin{center}
\textit{Can we improve our hardness amplification results for optimization problems?\\
Can we prove stronger uniform hardness amplification results for \MAXSAT?}
\end{center}

Arguably, optimization problems are as natural as decision problems, but are strictly harder from the point of view of computational complexity. Does this mean we can either give simpler proofs of hardness amplification for optimization problems or prove stronger results? We address the above questions in  Section~\ref{sec:resultNP}.

We now shift our focus to the class \P. As mentioned earlier, we have strong worst-case to average-case results established for problems in \P\ \cite{BRSV17,GR18}. The drawback however is that they are all for counting problems. This is indeed inherent as the underlying technique these works use are the same as the one used to show worst-case to average-case reduction for the permanent problem. While counting the number of $k$-cliques (the problem considered in \cite{GR18}) is a natural problem, and therefore hardness amplification for that problem is interesting, it still leaves the door open for proving hardness amplification for the search problem of just finding one $k$-clique in a graph (an easier problem and thus harder to amplify hardness).

\begin{center}
\textit{Can we prove hardness amplification results for natural search problems in \P?}
\end{center}

Moreover, there is a barrier \cite{Amir} to using the algebraic techniques of \cite{BRSV17,GR18} to obtain hardness amplification for important problems studied in fine-grained complexity such as computing the Longest Common Subsequence (\LCS) and Edit Distance (\Edit) for a pair of strings. In particular, if these string problems can be represented using low-degree polynomials, then we could obtain small speedups by using the polynomial method \cite{CW16}, which would imply new  circuit lower bounds \cite{AHWW16}. This suggests we might need to look beyond these algebraic techniques for proving hardness amplification for these string problems. Is there a different technique to prove hardness amplification in \P?
We address these aforestated questions in  Section~\ref{sec:resultP}.

\subsection{Our Results}\label{sec:results}

Our main contribution is a general hardness amplification theorem for optimization problems which we state in Section~\ref{sec:resultmain}. Next, we apply our main theorem to various problems. In Section~\ref{sec:resultNP} we state our hardness amplification theorems for various \NP-hard problems such as Knapsack and \MAXSAT.  In Section~\ref{sec:resultP} we state our hardness amplification theorems for various string problems in \P\ such as \LCS and \Edit Finally, in Section~\ref{sec:resultTFNP} we state our hardness amplification theorems for various  problems in \TFNP\ (believed to not be in \P) such as Factoring and computing Nash equilibrium. 
 
\subsubsection{Hardness Amplification of Optimization Problems}\label{sec:resultmain}

Aggregation is a key tool in the field of hardness amplification. If a function $f$ is hard to compute on a tiny fraction of the domain, then, intuitively, computing multiple instances of $f$ in one shot should be hard on a larger fraction of the inputs. More formally, for a function $f:[N]\to \Sigma$ and $k\in \N$, its $k$-direct product encoding is defined as a function $f^{(k)}:[N]^k\to \Sigma^k$ mapping  each tuple $(x_1,\dots, x_k)$ into $(f(x_1),\dots, f(x_k))$. Using standard techniques one can show a `direct product theorem' stating that if $f$ is hard against $t(n)$ running-time algorithms on $\alpha(n)$-fraction of the domain, then $f^{(k)}$ is hard against $t'(n)$ running-time algorithms on $\approx k\cdot \alpha(n)$-fraction of its domain. But in order to utilize such a direct product result, we need to be able to  stitch $k$-instances into a single (larger) instance. To address this task we introduce the following notion of direct product feasibility. 

\begin{definition}[Direct Product Feasibility; Informal statement  of Definition~\ref{def:DPfeasible}]\label{def:DPFeasible}
	
	Let $\Pi$ be an optimization problem. We say that $\Pi$ is $(S,T)$-direct product feasible\footnote{In the formal definition of direct product feasibility, it is defined for a pair of optimization problems $(\Pi,\Lambda)$ for technical reasons which will be addressed later in Section~\ref{sec:resultP}. In the case $\Pi=\Lambda$ we formally call it as self direct product feasible and this notion coincides with the informal definition given here. For most of the applications given in this paper, self direct product feasibility notion suffices.} if the exists a pair of deterministic algorithms $(\G, \Dec)$ satisfying the following:
\begin{itemize}
\item $\G$ takes as input $k$ instances $(I_1,\dots, I_k)$ of $\Pi$ each of size $n$ and  outputs an instance $I'$ of $\Pi$ of size $S(n,k)$.
\item $\Dec$ gets as input $(I_1,\dots, I_k)$, the instance $I'$ which is the output of $\G$ on input $(I_1,\dots, I_k)$,  an optimal solution for $I'$, and $i\in [k]$.  It outputs an optimal solution for the instance $I_i$. 
\item The running time of $\G$ and $\Dec$ is bounded by $T(n,k)$. 
\end{itemize}	
	 \end{definition}

Our main theorem is about hardness amplification for an arbitrary direct product feasible problem $\Pi$. In particular we show that if $\Pi$ is hard against $t(n)$ running time randomized algorithms on a tiny fraction of the domain, then $\Pi$ is hard on a much larger fraction of the domain against randomized algorithms with a similar running time.  

\begin{theorem}[Informal Statement of Theorem~\ref{thm:main}]\label{thm:intromain}
	Let $\Pi$ be $(S,T)$-direct product feasible. 
	Let $\D(n)$ be an efficiently samplable distribution over the instances of $\Pi$ of size $n$. Assume the following:
	\begin{itemize}
	\item Any $t(n)$ running-time algorithm with success probability at least $2/3$ fails to compute an optimal solution  on at least $\alpha(n)$-fraction of the inputs sampled from $\D$. 
	\item Fix $k=\poly((\alpha(n))^{-1})$. Then we have $T(n,k)=o(  t(n))$.
	\item We can (deterministically) decide the optimality of a given solution to any instance in $o(t(n))$ time.
	\end{itemize}
	
	Then there exists an efficiently samplable distribution $\D'(S(n,k))$ over instances of $\Pi$ of size $S(n,k)$ such that every $t(n)$ running-time algorithm with success probability at least $2/3$ fails to compute an optimal solution  on at least $99\%$ of the inputs sampled from $\D'$. 
\end{theorem}

Naturally, the distribution $\D'$ is defined as follows: Draw $k$ independent samples $I_1,\dots, I_k$ from $\D$, and output $\G(I_1,\dots, I_k)$. 
The proof of our main theorem is based on a reduction using an oracle access to an algorithm that solves $\D'$ on 99\% of the inputs, we convert it into an algorithm solving $\D$ on greater than $1-\alpha(n)$ fraction of the inputs. The reduction is uniform, so in case that the algorithms $(\G, \Dec)$ are uniform we get a uniform hardness amplification result. 

Another key point is that our hardness amplification is a self-reduction, i.e.,  if a problem is somewhat hard against one distribution $\D$, then the same problem is much harder against a different distribution $\D'$.

To the best of our knowledge, this is the first result to study hardness amplification for optimization problems. It opens avenues to prove results in various subclasses as we will see in subsequent subsections.

\subsubsection{Hardness Amplification for \NP-hard Problems}\label{sec:resultNP}

In the \NP\ world, we generalize the results of~\cite{O04,Tre05} to optimization problems. In particular we show that if \MAXSAT is hard to solve on $\nicefrac{1}{\poly(n)}$ fraction of the inputs of samples drawn from some samplable distribution $\mathcal{D}$. Then there exists a samplable distribution $\mathcal{D'}$ such that solving \MAXSAT on $\mathcal{D'}$ is hard on at least $\nicefrac {99}{100}$-fraction of the samples (See Corollary~\ref{cor:SATpoly}).

\begin{theorem}[Informal Statement of Corollary~\ref{cor:SATpoly}.]\label{thm:SATintro}
	Let $\D(n)$ be a distribution over  $3$-\CNF formulas with $n$ variables and $\poly(n)$ clauses,  such that for every randomized algorithm $\mathcal{A}$ running in time $\poly(n)$, we have:
	$$
	\Pr_{\Psi\sim \D}\left[\mathcal A \text{ finds a optimal assignment for } \Psi \text{ w.p.}\ge 2/3\right]\le 1-1/\poly(n).
	$$
	
Then there exists a distribution $\D'(n')$ over $3$-\CNF formulas with $n'$ variables $\poly(n')$ clauses,  such that for every randomized algorithm $\mathcal{A}'$ running in time $ \poly(n')$, we have:
	$$
	\Pr_{\Psi'\sim \D'}\left[\mathcal A' \text{ finds a optimal assignment for }\Psi' \text{ w.p.}\ge 2/3\right]\le 0.01.
	$$
	 
	Moreover, if $\D(n)$ is $\poly(n)$-samplable then $\D'(n')$ is $\poly(n')$-samplable.
	
\end{theorem}

Observe that the failure probability on $\mathcal{D'}$ is much larger than in~\cite{O04,Tre05} and can even tend to $0$ for a proper choice of our parameters. This can be achieved since we deal with optimization problems instead of decision problems. 

We also remark that our reduction and the proof correctness are much simpler, and in particular we do not rely the hard core set lemma \cite{I95}, a powerful and non-trivial key tool in the previous known proofs. 

Our result easily extends  into other \NP-hard problems such as finding the largest clique in a graph, or finding smallest dominating set or vertex cover of a graph, etc (see Remark~\ref{rem:maxprobNP}).

However, there are other \NP-hard problems for which establishing a hardness amplification result through Theorem~\ref{thm:intromain} is not easy. A special highlight is that of proving such a result for the  Knapsack problem, as it isn't immediately clear if it's direct product feasible for reasonable range of parameters. This is because,  for the Knapsack problem, when we aggregate instances in the natural way, optimal solutions of one instance may interfere with other instances (see Section~\ref{sec:NP} for details). Nonetheless, with some care, the direct product feasibility of Knapsack problem was established (see Lemma~\ref{lem:sdfKS}).

The Exponential Time Hypothesis (\ETH{}) \cite{IP01,IPZ01,CIP06} asserts that that we cannot decide whether a given $3$-\CNF is satisfiable in time which is sub-exponential in the number of variables. That is a worst case assumption, and it raises a natural question arises: Can we prove stronger hardness amplification result based on \ETH{}? In fact, can we prove a worst case to an average case hardness amplification based on \ETH{}? 

Our next theorem is a step towards proving such a worst-case to an average case reduction for \MAXSAT.
\begin{theorem}[Informal Statement of Corollary~\ref{cor:SATETH}.]\label{thm:maxSATETHInformal}
	Let $\D(n)$ be a distribution over $3$-\CNF instances with $n$ variables and $O(n)$-clauses,  such that for every randomized algorithm $\mathcal{A}$ running in time $2^{o(n)}$, we have:
	$$
	\Pr_{\Psi\sim \D}\left[\mathcal A \text{ finds an optimal assignment for } \Psi \text{ w.p.}\ge 2/3\right]\le 1-\frac{1}{2^{o(n)}}.
	$$
	
	Then there exists a distribution $\D'(n')$ over $3$-\CNF instances with $n'$ variables and $2^{o(n')}$ clauses, such that for every polynomial time randomized algorithm $\mathcal{A}'$, we have:
	$$
	\Pr_{\Psi' \sim \D'}\left[\mathcal A' \text{ finds an optimal assignment for  }\Psi' \text{ w.p.}\ge 2/3\right]\le 0.01.
	$$
\end{theorem}

Heally et al.~\cite{HVV06} proved a similar result for the non-uniform case. Our result is stronger in the sense that we use a weaker assumption: we rely on the \ETH{} that is assuming that every \textit{uniform} algorithm fails on $1/2^{o(n)}$-fraction of inputs. While Heally et al. use similar assumption against non-uniform algorithms.

\subsubsection{Hardness Amplification in \P}\label{sec:resultP}

We investigate hardness amplification in \P\ and can show results for string problems, such as \LCS\ and \Edit, which were not possible in previous works.

\begin{theorem}[Informal statement; see Corollaries~\ref{cor:LCS}~and~\ref{cor:Edit} ]\label{thm:string}
Fix $\varepsilon>0$.	Let $\D(n)$ be an efficiently samplable distribution over the instances of \LCS/\Edit of length $n$. Assume that any $n^{2-\varepsilon}$ running-time algorithm with success probability at least $2/3$ fails to compute an optimal alignment  on at least $1/n^{o(1)}$-fraction of the inputs sampled from $\D$. 
	Then for some $\varepsilon'>0$ there exists an efficiently samplable distribution $\D'(n^{1+o(1)})$ over instances of \LCS/\Edit of size $n^{1+o(1)}$ such that every $n^{2-\varepsilon'}$ running-time algorithm with success probability at least $2/3$ fails to compute an optimal solution  on at least $99\%$ of the inputs sampled from $\D'$. 
\end{theorem}

Recall from earlier in this section that Abboud \cite{Amir} had pointed out a barrier to obtaining a result such as above, through algebraic techniques. Another similarity search problem that is studied along with \LCS and \Edit, is the problem of computing the Fr\'echet distance between two (discrete) curves. Strangely, this problem resists all natural approaches to show that it is direct product feasible (see Remark~\ref{rem:frechet}). Therefore, it is an interesting question as to whether it is possible to show that it is direct product feasible (for relevant range of parameters) or whether it is a candidate for a problem that is not direct product feasible. 

Additionally, we show hardness amplification for a very different kind of problem, that of computing the product of two matrices (see Corollary~\ref{cor:matmult}). We highlight this problem, as it does not directly follow from our main theorem (i.e., Theorem~\ref{thm:intromain}). Elaborating, a detail that was brushed under the carpet while discussing Theorem~\ref{thm:intromain} was that, given an instance of an optimization problem and a candidate solution, we need to able to efficiently compute the value of the objective of the candidate solution for that instance. This naturally holds for all the problems considered in this paper except the task of computing the product of two matrices, i.e., we do not know a way to \emph{deterministically} verify if the   product of two matrices is equal to a given third matrix, which is significantly faster than actually multiplying the two given matrices and checking if it's equal to the third matrix \cite{K18,WW18}. Nonetheless, we modify the proof of Theorem~\ref{thm:intromain} to handle this issue.

\subsubsection{Hardness Amplification in \TFNP}\label{sec:resultTFNP}

Total problems (with not necessarily efficient verification of totality) are essentially equivalent to Optimization problems (see Remark~\ref{rem:total}). The class \TFNP\ is special as it is in an informal sense the intersection of Search \NP\ and Optimization problems. Problems in \TFNP\ capture problems in various areas such as game theory, cryptography, computational geometry, etc. We show that our general theorem can be applied to \TFNP\ problems as well, and as an example show it for the Factoring problem (see Corollary~\ref{cor:factor}) and the End of a Line problem (see Corollary~\ref{cor:EOLpoly}). The latter hardness amplification result directly implies the hardness amplification of various problems in game theory such as computing an approximate Nash equilibrium (see Section~\ref{sec:EOL} for details).

\subsection{Open Problems}

Our work leaves open several questions. We state a few of them below.

\subsubsection{Stronger Hardness Amplification}\label{sec:strongopen}
In Theorem~\ref{thm:maxSATETHInformal} we showed that if \MAXSAT is hard to compute on $1-1/2^{o(n)}$-fraction of inputs for  sub-exponential time algorithm, then there exists a distribution on which it is hard on a constant fraction of inputs for algorithms running in time $n^{\omega(1)}$. A natural open question is the following:

\begin{center}
\textit{Can we improve Theorem~\ref{thm:maxSATETHInformal} and get hardness amplified against sub-exponential time algorithms (instead of super-polynomial time algorithms)?}
\end{center}

It seems to us that derandomized direct product theorems may serve as the key tool to address the above question (for example, see \cite{IKW12}). In particular, if one can prove a (strongly) derandomized version of~\cite{FK00} then it might be possible to both aggregate sub-exponentially many instances succinctly and sample from the (derandomized) direct product distribution efficiently.

\subsubsection{Direct Product Feasibility}
In this paper, we were able to show direct product feasibility for certain problems quite easily (for example, see Theorems~\ref{thm:SATintro}~and~\ref{thm:string}), but had to work harder to prove them for some other problems (for example, see Lemmas~\ref{lem:sdfKS}~and~\ref{lem:MatMult}), and in some problem(s) were unable to establish the property of direct product feasibility (see Remark~\ref{rem:frechet}). This leads us to the following question.

\begin{center}
\textit{
Can we pinpoint what property of a problem makes it possible to establish \\
direct product feasibility?}
\end{center}

\subsubsection{Gap Amplification versus Hardness Amplification}

Direct Product theorems are key ingredients for both gap amplification and hardness amplification. Also, there are many philosophical similarities in the techniques known in literature of the aforementioned two kinds of amplifications. Thus we can ask the following (ambitious) question:

\begin{center}
\textit{
Can we obtain a trade-off between gap amplification and hardness amplification?}
\end{center}

 In particular, can we show that if one problem is hard to approximate on worst case within some factor $\alpha>0$, then it is hard to approximate within a factor $\alpha/100$ on average? We note here that Feige \cite{F02}, did answer the converse of this question, i.e., he used average case hardness assumptions to prove hardness of approximation results for various problems in \NP.

It seems to us that analyzing the operation of performing a small perturbation on the given instance may be the right direction to proceed.  Elaborating, consider a (worst case) hard   distribution over gap instances of some problem. If we build a new distribution, which samples from the aforementioned distribution, then performs  a small perturbation on the sampled gap instance, and outputs the perturbed instance, then we would still retain most  of the gap in the instance sampled from the new distribution, but on the other hand, the fraction of instances on which it is hard to solve the problem should increase significantly. It would be interesting if this intuition/approach could be made to work. 

A related question is to ask if we can improve our result in Theorem~\ref{thm:maxSATETHInformal} (for example, by making progress on the question detailed in Section~\ref{sec:strongopen}) using Gap-ETH \cite{D16,MR16} (instead of ETH)?

\subsubsection{Average Case Hard Problems in \P}

In this paper, we looked at average case hardness of some problems in \P\ against some efficiently sampleable distribution but one can ask if we can achieve more.  

\begin{center}
\textit{
Can we show for some natural problem in \P\ that it is hard to solve for the \\ uniform distribution?}
\end{center}

Another important question stemming from cryptography \cite{BRSV17} is whether we can construct a \emph{fine-grained} one way function from worst case assumptions? 

\subsection{Organization of Paper}

In Section~\ref{sec:overview}, we provide the proof overview of our main theorem (Theorem~\ref{thm:intromain}). In Section~\ref{sec:main}, we formally state and prove Theorem~\ref{thm:intromain}. In Section~\ref{sec:NP}, we prove hardness amplification results for various problems in \NP.
In Section~\ref{sec:P}, we prove hardness amplification results for various problems in \P.
Finally, in Section~\ref{sec:TFNP}, we prove hardness amplification results for various problems in \TFNP.

\section{Proof Overview}\label{sec:overview}
We provide a proof overview for our hardness amplification result for the problem of finding the maximum clique in a graph and then in the subsequent section we will show how our general result (i.e., Theorem~\ref{thm:intromain})  would follow.
%Finally, we briefly discuss how to obtain the various specific hardness amplification results for problems mentioned in Section~\ref{sec:results}.

\subsection{Hardness Amplification for Max Clique}
To illustrate the main ideas behind our scheme let us focus on  \MAXCLIQUE,  the problem of finding the largest clique in a given graph $G$.

Assume the existence of a distribution  $\mathcal{D}$ over graphs on $n$ vertices which is somewhat hard to compute. That is for every  \emph{randomized} algorithm $\mathcal{A}$ running in time $\poly(n)$, we have
\begin{eqnarray}\label{eq:A}
	\Pr_{G\sim \mathcal{D}}\left[\mathcal A \text{ finds max-clique in }G \text{ w.p.}\ge 2/3\right]\le 1-1/n.
\end{eqnarray}

We would like to prove the existence of a new distribution $\mathcal{D}'$ over graphs on $\poly(n)$ vertices which is much harder to compute. That is, for every randomized algorithm $\mathcal{A}'$ running in time $\poly(n')$, we have:
\begin{eqnarray}\label{eq:A'}
	\Pr_{G'\sim \mathcal{D}'}\left[\mathcal A' \text{ finds max-clique in }G' \text{ w.p.}\ge 2/3\right]\le 0.01.
\end{eqnarray}

Moreover if $\mathcal{D}$ is $\poly(n)$-time samplable, then so is $\mathcal{D'}$.
\paragraph{Construction of New Distribution:}
$\mathcal{D}'$ samples a graph $H$ as follows:
\begin{enumerate}
	\item Independently sample $G_1,\ldots ,G_k$ from $\mathcal{D}$, where $k=\poly(n)$.
	\item Define $V(H)=V(G_1)\dot\cup \cdots \dot\cup V(G_k)$.
	\item For every $i\in[k]$, connect the vertices in $V(G_i)$ using the original edges in $G_i$.
	\item For every $i,j\in[k]$ such that $i\neq j$, insert all the possible edges between $G_i$ and  $G_j$.
	\item Output $H$.
\end{enumerate}
Clearly, if $\mathcal{D}$ is $\poly(n)$-time samplable, then so is $\mathcal{D'}$.
Now assume for sake of contradiction, that there exists $\mathcal{A}'$ running in time $\poly(n')$, violating Equation~(\ref{eq:A'}).
We show the existence of an algorithm $\mathcal{A}$ running in time $\poly(n)$ violating Equation~(\ref{eq:A}).

The algorithm $\mathcal{A}$ on input graph $G$ with $n$ vertices is defined as follows:
\begin{enumerate}
	\item Let $\S$ be an empty set.
	\item Repeat following $O(n)$ times.
	\begin{enumerate}
		\item Pick randomly $i\in [k]$.
		\item Independently sample $G_1,\ldots ,G_{i-1},G_{i+1},\ldots G_k$ from $\mathcal{D}$.
		\item Construct $H$ setting $G_i$ to be $G$.
		\item Find clique in $H$ using $\mathcal{A}'$.
		\item Restrict clique in $H$ to the vertices of $G$ and add it to $\S$.
	\end{enumerate}
	\item Output the largest clique in $\S$.
\end{enumerate}
Clearly, the running time of $\mathcal{A}$ is $\poly(n)$, as $n'=\poly(n)$ and the running time of $\mathcal{A'}$ is $\poly(n')$. Our first observation is that for any graph $H$ constructed by $\mathcal{A}$, and for every $i\in [k]$ the restriction of a maximal clique in $H$ into $G_i$, is a maximal clique for $G_i$.

Let $\mathcal{A}_0$ be one iteration of step 2 of $\mathcal{A}$. If we show that $\mathcal{A}_0$ outputs maximum clique w.p. $\Omega(1/n)$ on  $1-1/n$ fraction of samples from $\mathcal{D}$ then,
$\mathcal{A}$ outputs maximum clique w.p. $2/3$ on $1-1/n$ fraction of samples from $\mathcal{D}$.

Now, observe that if instead of planting the given input graph $G$ as the $i$-th subgraph of $H$, we were planting a uniformly random sample of $\mathcal{D}$, then we get a graph $H$ which is drawn according to $\mathcal{D'}$. Consequently, if that was the case, then the success probability of $\mathcal{A}_0$ was equal the probability of $\mathcal {A'}$ and we were done. 

Let $\mathcal{D}'_G$ denote the marginal distribution over $H$, where the graph $G$ is planted at a random coordinate $i\in [k]$. We conclude the proof by showing that for $1-1/n$-fraction of instances $G$ drawn from $\mathcal{D}$ we have: 
$$
	\Pr_{G'\sim \mathcal{D}'_G}\left[\mathcal A' \text{ finds max-clique in }G' \text{ w.p.}\ge 2/3\right]\ge\frac{1}{2} \Pr_{G'\sim \mathcal{D'}}\left[\mathcal A' \text{ finds max-clique in }G' \text{ w.p.}\ge 2/3\right].
$$
Towards this goal we use a result by Feige and Kilian~\cite{FK00} that was proven in the context of parallel repetition. Under minor manipulations their result can be stated as follows:

Let $X$ be a universe and $\mathcal{T}$ be a distribution over $X$. Let $f:X^k\to\{0,1\}$. 
Define $$\mu=\underset{x^k\sim \mathcal{T}^k}{\mathbb{E}}\left[f\left(x^k\right)\right],$$ 
$$\mu_{x}=\underset{i\in [k], x_1,\ldots,x_{i-1},x_{i+1},\ldots ,x_k\sim \mathcal{T}}{\mathbb{E}}\left[f(x_1,\ldots,x_{i-1},x,x_{i+1},\ldots x_k)\right].$$
\begin{eqnarray}\label{eq:FK}
\Pr_{\substack{x\sim \mathcal{T}}}\left[|{\mu_{x}-\mu}|\ge k^{-1/6}\right]\le k^{-1/6},
\end{eqnarray}

To conclude the result, set $X$ as the set of graphs with $n$ vertices, and $\mathcal{T}$ be the distribution $\mathcal{D}$. We have $\mathcal{D'}=\mathcal{D}^k$. Define $f:X^k\to \zo$ by:
$$ 
f(G')=1 \iff \mathcal{A'} \text{ finds a maximal clique in $G$ w.p.} \ge 2/3. 
$$
In these notations, 
\begin{eqnarray*}
	\mu &=& \Pr_{G'\sim \mathcal{D'}}\left[\mathcal A' \text{ finds max-clique in }G' \text{ w.p.}\ge 2/3\right]\\
	\mu_x &=& \Pr_{G'\sim \mathcal{D}'_G}\left[\mathcal A' \text{ finds max-clique in }G' \text{ w.p.}\ge 2/3\right]. 
\end{eqnarray*}
 By an application of (\ref{eq:FK}), and a proper choice of $k$, we get that for all but at most $k^{1/6}$-fraction of graphs $G$ drawn according to $ \mathcal{D}$,  the success probability of $\mathcal{A'}$ on $\mathcal{D'}_G$ is $\Omega(n)$, as claimed.
\subsection{Abstraction}

In the previous subsection, we showed the main ingredients used for proving hardness amplification for the task of finding a maximal clique in a given graph. What were the properties of \MAXCLIQUE that we utilized to prove the result? 

One property that we used was that if we are given $k$ input graphs $G_1,\dots, G_k$, there exists an efficient way to construct a large graph $H$ such that a maximal clique in $H$ induces a maximal clique on each of the graphs $G_i$. The second property was that given a maximal clique in $H$ there exists an efficient algorithm to construct a maximal clique on each of the graphs $G_i$.

These two properties are captured in Definition~\ref{def:DPFeasible}: The first property of a problem $\Pi$ being Direct Product feasible is the existence of an efficient algorithm $\G$ stitching $k$ instances $I_1,\dots, I_k$ of $\Pi$ into a larger instance $I'$ of $\Pi$, such that: an optimal solution for $I'$ induces an optimal solution for each of the instances $I_i$. The second property the existence of an efficient algorithm $\Dec$ converting an optimal solution for $I'$ into an optimal solution of $I'$.

Once we show $\Pi$ is Direct Product feasible then the rest of the proof goes through. Indeed, assuming the existence of a distribution $\mathcal{D}$ on instances of $\Pi$ for which any efficient algorithm fails to compute on $1-1/n$ fraction of inputs,
we define the distribution $\mathcal{D'}, \mathcal{D'}_I$ as follows:
\begin{itemize}
	\item $\mathcal{D'}$ is the $k$-product distribution of $\mathcal {D}$, where we pick $k$ random samples from $\mathcal{D'}$ independently.
	\item $\mathcal{D'}_I$ is the distribution where we pick uniformly at random $i\in [k]$, and independently sample $I_1,\ldots ,I_{i-1},I_{i+1},\ldots I_k$ from $\mathcal{D}$. Finally, we construct $I'$ by setting $I_i$ to be $I$. 
\end{itemize} 

Now we can use~\cite{FK00} to show that for most instances $I\sim \mathcal{D}$ to connect the success probability of $\mathcal{A'}$ on $\mathcal{D'}$ and $\mathcal{D'}_I$, to conclude the proof.

\paragraph{Remark about Direct Product results and Hardness Amplification.} The direct product lemma at the heart of most hardness amplification results is the XOR lemma \cite{Y82}. But here we critically use the fact the problem is total, so at the surface at least, our results are incomparable to the hardness amplification results for \NP\ and \EXP\ obtained via XOR lemmas. 
%
%\subsection{Application to Other Avenues}
%
%\subsubsection{Hardness Amplification for  \NP-hard Optimization Problems}
%
%
%\subsubsection{Hardness Amplification in \P}
%
%\subsubsection{Hardness Amplification in \TFNP}

\section{Hardness Amplification of Optimization Problems}\label{sec:main}

In this section, we prove our main result. First, we define some notations.
We use the   definition of optimization problems given in \cite{ACGKMP99} with additional formalism.

\begin{definition}[Optimization Problem]
An optimization problem $\Pi$ is charaterized by the following quadruple of objects $(I_\Pi,\SOL_{\Pi},\m_\Pi,\goal_\Pi)$, where:
\begin{itemize}
\item $I_\Pi$ is the set of instances of $\Pi$. In particular for every $d\in\mathbb{N}$, $I_\Pi(d)$ is the set of instance of $\Pi$ of input size at most $d$ (bits); 
\item $\SOL_\Pi$ is a function that associates to any input instance $x\in I_\Pi$ the set of feasible solutions of $x$;
\item $\m_\Pi$ is the measure  function\footnote{We define the measure function only for feasible solutions of an instance. Indeed if an algorithm solving the optimization problem outputs a non-feasible solution then, the measure just evaluates to -1 in case of maximization problems and $\infty$ in case of minimization problems.}, defined for pairs $(x,y)$ such that $x\in I_\Pi$ and $y\in \SOL_\Pi(x)$. For every such pair $(x,y)$,  $\m_\Pi(x,y)$ provides a non-negative integer which is the value of the feasible solution $y$; 
\item $\goal_\Pi\in\{\min,\max\}$ specifies whether $\Pi$ is a maximization or minimization problem.
\end{itemize}
\end{definition}

We would like to identify a subset of our solution space which are optimal with respect to our measure function. To this effect, we define a notion of optimal feasible solution.

\begin{definition}[Optimal Feasible Solution]
Let $\Pi(I_\Pi,\SOL_{\Pi},\m_\Pi,\goal_\Pi)$ be an optimization problem. For every $x\in I_\Pi$ and $y\in \SOL_\Pi(x)$ we say that $y$ is an optimal feasible solution of $x$ if for every $y'\in \SOL_\Pi(x)$ we have $\m_\Pi(x,y)\ge \m_\Pi(x,y')$ if $\goal_\Pi=\max$ and  $\m_\Pi(x,y)\le \m_\Pi(x,y')$ if $\goal_\Pi=\min$.
\end{definition}

Now, we can formally define the notion of direct product feasibility of a pair of optimization problems. 

\begin{definition}\label{def:DPfeasible}
Let $\Pi(I_\Pi,\SOL_{\Pi},\m_\Pi,\goal_\Pi)$ and $\Lambda(I_\Lambda,\SOL_{\Lambda},\m_\Lambda,\goal_\Lambda)$ be two optimization problems. Let $S,T:\mathbb{N}\times \mathbb{N}\to\mathbb{N}$. We say that the pair $(\Pi,\Lambda)$ is $(S,T)$-direct product feasible if there exists a pair of deterministic  algorithms $(\G,\Dec)$ such that for every $k,d\in\mathbb{N}$ the following holds:
\begin{itemize}
\item $\G$ takes as input $x_1,\ldots ,x_k\in I_\Pi(d)$ and outputs $x'\in I_\Lambda(S(d,k))$.
\item For any  feasible solution $y'\in \SOL_\Lambda(x')$, $\Dec$ takes as input $i\in[k]$, $x_1,\ldots , x_k\in I_\Pi(d)$, and $y'$, and outputs $y\in \SOL_\Pi(x_i)$. Moreover if  $y'\in \SOL_\Lambda(x')$ is an optimal feasible solution then so is  $y\in\SOL_\Pi(x_i)$.
\item $\G$ and $\Dec$ run in $T(d,k)$ time.
\end{itemize} 
Moreover, if $\Pi=\Lambda$ then we say that $\Pi$ is $(S,T)$-self direct product feasible
\end{definition}

All but two results in this paper use the notion of self direct product feasibility. Only the problems of computing the longest common subsequence and computing the edit distance between two strings require us to define direct product feasibility for a pair of problems (instead of a single problem). Even for the aforementioned two problems, direct product feasibility is shown for a pair of essentially \emph{same} problems (see Lemmas~\ref{lem:DPLCS}~and~\ref{lem:DPEdit} for more details).  

We now show our main theorem, that direct product feasibility implies hardness amplification.

\begin{theorem}[Formal version of Theorem~\ref{thm:intromain}]\label{thm:main}
\begin{sloppypar}Let $p\in(0,1)$.
Let $\Pi(I_\Pi,\SOL_{\Pi},\m_\Pi,\goal_\Pi)$ and $\Lambda(I_\Lambda,\SOL_{\Lambda},\m_\Lambda,\goal_\Lambda)$ be two optimization problems. Let $S,T:\mathbb{N}\times \mathbb{N}\to\mathbb{N}$ be such that $(\Pi,\Lambda)$ is $(S,T)$-direct product feasible. Let $v:\mathbb{N}\to\mathbb{N}$ be such that there is a deterministic algorithm $\mathcal{V}$ running in time $v(d)$ which on input $x\in I_\Pi(d)$ and $y\in \SOL_\Pi(x)$ always correctly computes $\m_\Pi(x,y)$. Let $s,t:\mathbb{N}\to\mathbb{N}$, $\p:\mathbb{N}\to (0,1]$, and  $\mathcal{D}=\{D_d\}_{d\in\mathbb{N}}$ be a family of distributions such that for every randomized algorithm $\mathcal{A}$ running in time $t(d)$, the following holds for large $d\in\mathbb{N}$.
\begin{itemize}
\item $D_d$ is a distribution over $I_\Pi(d)$ where an instance of $I_\Pi(d)$ can be sampled from $D_d$ in $s(d)$ time.
\item $
\underset{x\sim D_d}{\Pr}\left[\mathcal A \text{ finds an optimal feasible solution for }x \text{ with probability at least } p\right]\le 1-\p(d)$.
\end{itemize}\end{sloppypar}

Let $k:=64\cdot( \p(d))^{-6}$ and $c:=\frac{200\ln\left(\nicefrac{1}{1-p}\right)}{p}$. If
$k\cdot s(d)+T(d,k)+v(d)\le \frac{t(d)}{2c}$, then there is a distribution family $\mathcal{D}'=\{D_d'\}_{d\in\mathbb{N}}$ such that for every randomized algorithm $\mathcal{A}'$ running in time $\frac{t(d)}{2c}$, the following holds for all large enough $d\in\mathbb{N}$.
\begin{itemize}
\item $D_d'$ is a distribution over $I_{\Lambda}(m)$  where $m=S(d,k)$ and an instance in $I_\Lambda(m)$ can be sampled from $D_d'$ in $O(s(d)k+T(d,k))$ time.
\item $
\underset{x'\sim D_d'}{\Pr}\left[\mathcal A' \text{ finds an optimal feasible solution for }x' \text{ with probability at least }p\right]\le 0.01.
$
\end{itemize}
  \end{theorem}

A key ingredient in the proof of the above theorem is the following direct product lemma, which may be seen as a combinatorial analogue of a special case of the celebrated parallel repetition theorem \cite{R98}:
\begin{lemma}[Feige-Kilian Direct Product Lemma \cite{FK00}]\label{lem:DP}
Let $X$ be a finite set and $\D$ a probability distribution on $X$. Let $k\in\mathbb{N}$ and $f:X^k\to\{0,1\}$. We define the following two measures:
$$
\mu=\underset{x_1,\ldots ,x_k\sim \D}{\mathbb{E}}\left[f(x_1,\ldots ,x_k)\right]=\underset{x_1,\ldots ,x_k\sim \D}{\Pr}\left[f(x_1,\ldots ,x_k)=1\right],
$$
where $x_1,\ldots ,x_k$ are sampled independently, and for every $i\in[k]$ and $x\in X$ we define 
$$
\mu_{i,x}=\underset{x_1,\ldots ,x_{i-1}, x_{i+1},\ldots,x_k\sim \D}{\mathbb{E}}\left[f(x_1,\ldots.x_{i-1},x,x_{i+1},\ldots ,x_k)\right],
$$
where $x_1,\ldots,x_{i-1},x_{i+1},\ldots ,x_k$ are sampled independently. Then, we have the following:
$$
\Pr_{\substack{x\sim \D\\i\sim [k]}}\left[|\mu_{i,x}-\mu|\ge \frac{1}{k^{1/3}}\right]\le \frac{1}{k^{1/3}},
$$
where we sample $i$ from $[k]$ uniformly at random.
\end{lemma}
The proof of the above lemma as stated above may be found in \cite{O05} and we provide it in Appendix~\ref{sec:missing} for completeness. We also note that variants of the above lemma have previously appeared in direct product testing literature \cite{DG08,IKW12,DS14}.

\begin{proof}[Proof of Theorem~\ref{thm:main}]
First we describe the construction of $\mathcal{D}'$ and then show the claims made in the theorem statement. 
\subsection*{Construction of $\mathcal{D}'$} Let $(\G,\Dec)$ be the pair of algorithms guaranteed by Definition~\ref{def:DPfeasible} for the pair $(\Pi,\Lambda)$. Fix $d,k\in\mathbb{N}$. We construct $D_d'$ from $D_d$ as follows. Independently sample $k$ pairs of instances, say $x_1,\ldots ,x_k$ from $D_d$  and feed it as input to $\G$. 
The sampling algorithm of the distribution $D_d'$ then outputs the output of $\G$.

Therefore, the $m$ in the theorem statement, the size of the instances outputted by $D_d'$ is equal to $S(d,k)$. The time needed to sample from $D_d'$ is the time needed to sample $k$ independent samples from $D_d$ which is $k\cdot s(d)$ time, plus the running time of $\G$ which is $T(d,k)$. 

\subsection*{Correctness of the Claim}  We will show that if there is a randomized algorithm $\mathcal{A}'$ with success probability $p$  running in time $\nicefrac{t(d)}{2c}$ that finds an optimal feasible solution of an instance sampled from $D_d'$ with probability (over the sampling) greater than $0.01$ then, there is a randomized  
algorithm $\mathcal{A}$ with success probability $p$ running in time $t(d)$ that finds an optimal feasible solution for an instance sampled from $D_d$ with probability (over the sampling) greater than $1-\p(d)$, reaching a contradiction. 
First we describe below  the algorithm\footnote{We remark here that the simulation of instances sampled from $D_d'$ from an instance of $D_d$ is similar to the simulation described in the (textbook) proof of showing existence of weak one-way functions imply existence of strong one-way functions \cite{Y82,G09}.} $\mathcal{A}$.  

\begin{center}\begin{tcolorbox}[breakable,colback=gray!10!white,colframe=black,width=0.9\textwidth]
	\noindent\textbf{Algorithm $\mathcal{A}$}:\vspace{0.1cm}
			\par
			\noindent \textbf{Input}: An instance $x\in I_{\Pi}(d)$.\vspace{0.1cm}
			\par
			\noindent \textbf{Output}: A feasible solution $y\in\SOL_{\Pi}(x)$.\vspace{0.1cm}
			\par
			\noindent \textbf{Procedure}: 
	\begin{enumerate}
				\item Let $\mathcal{S}\subseteq\SOL_\Pi(x)$ be a subset of feasible solutions initialized to  $\emptyset$.
				\item Repeat the below procedure $c$ times.
				\begin{enumerate}
				\item[2.1.] Pick $i\in[k]$ uniformly at random. 
				\item[2.2.] Independently sample $k-1$ pairs from $D_d$ say $x_1,\ldots ,x_{i-1},x_{i+1},\ldots ,x_k$.
				\item[2.3.] Define $x_i=x$. 
				\item[2.4.] Feed $x_1,\ldots ,x_k$ as input to $\G$. Let $x'\in I_\Lambda(m)$ be the output of $\G$.
				\item[2.5.] Feed $x'$ as input to $\mathcal{A}'$. Let $y'\in\SOL_{\Lambda}(x')$ be the output of $\mathcal{A}'$.
				\item[2.6.] Feed $i$, $x_1,\ldots ,x_k$, and $y'$ as input to $\Dec$. Let $y$ be the output of $\Dec$.
				\item[2.7.] Include $y$ in $\mathcal{S}$.
				\end{enumerate}
				\item Run $\mathcal{V}$ on $(x,y)$ for each $y\in\mathcal{S}$ and output the feasible solution in $\mathcal{S}$ which optimizes $
			\Pi$ (depends on $\goal_\Pi$). 
			\end{enumerate}\end{tcolorbox}
\end{center}

Let us first analyze the running time of $\mathcal{A}$. Since we repeat Step 2, $c$ times, it suffices to analyze the time needed for one iteration. Step 2.2 needs $(k-1)\cdot s(d)$ time, Step 2.4 needs $T(d,k)$ time, Step 2.5 needs time $\nicefrac{t(d)}{2c}$ time, and finally Step 2.6 needs time $T(d,k)$. Therefore, the total running time of  $\mathcal{A}$ is less than $c(k\cdot s(d)+T(d,k)+\nicefrac{t(d)}{2c}+v(d))\le t(d)$.

Finally, we argue on the correctness probability of $\mathcal{A}$. Let $\mathcal{A}_0$ be the same algorithm as $\mathcal{A}$ except Step 2 has only one iteration. Therefore it suffices to show that $\mathcal{A}_0$ outputs an optimal feasible solution with probability at least $\frac{\ln\left(\nicefrac{1}{1-p}\right)}{c}$ on $1-\p(d)$ fraction of instances sampled from $D_d$, as this implies that  
$\mathcal{A}$ outputs an optimal feasible solution with probability at least $\left(1-\left(1-\frac{\ln\left(\nicefrac{1}{1-p}\right)}{c}\right)^c\right)\ge 1-e^{\ln{(1-p)}}=p $ on greater than $1-\p(d)$ fraction of instances sampled from $D_d$. Notice that if $y'$ is an optimal solution to $x'$ then $\Dec$ always outputs an optimal solution to $x$. Therefore, it suffices to show that on input $x'$, $\mathcal{A}'$ in Step 2.5. outputs an optimal feasible solution  with probability at least $\frac{\ln\left(\nicefrac{1}{1-p}\right)}{c}$ on greater than $1-\p(d)$ fraction of instances sampled from $D_d$. 

Consider the Boolean function $f:I_\Lambda(m)\to\{0,1\}$ where $f(x')=1$ if and only if $\mathcal{A}'$ outputs an optimal feasible solution of $x'$ with probability at least $p$. From assumption on the fraction of sampled inputs that $\mathcal{A}'$ outputs an optimal feasible solution on, we have that $\mu:=\underset{x'\sim D_d'}{\mathbb{E}}\left[f(x')\right]> 0.01$.  
Define $\mu_{x,i}$ as follows:
$$\mu_{x,i}:= \Pr_{x_1,\ldots,x_{i-1},x_{i+1},\ldots ,x_k\sim D_d}[f(x_1,\ldots,x_{i-1},x,x_{i+1},\ldots ,x_k)=1].$$ From Lemma~\ref{lem:DP}, we have the following, 

$$\underset{\substack{x\sim D_d\\ i\in[k]}}{\Pr}\left[\left|\mu_{x,i}-\mu\right|\ge k^{-1/3}\right]\le k^{-1/3}.$$

Call an instance $x\in I_{\Pi}(d)$  ``bad'' if  $\underset{ i\in[k]}{\Pr}\left[\mu_{x,i}<\mu- k^{-1/6}\right]\le k^{-1/6}$; otherwise it is called ``good''. From Markov inequality we have that: 

$$
\Pr_{x\sim D_d}[x \text{ is good}] \ge 1-k^{-1/6}.
$$

\begin{sloppypar}Take any good $x$ then with probability at least $1-k^{-1/6}$ over $i\in [k]$ we have that $\mu_{x,i}\ge \mu -k^{-1/6}\ge 0.01-k^{-1/6}$. Next, conditioned on picking $i\in [k]$ such that $\mu_{x,i}>0.01-k^{-1/6}$, then with probability at least $0.01-k^{-1/6}$ we have $f(x_1,\ldots,x_{i-1},x,x_{i+1},\ldots ,x_k)=1$. Conditioned on $f(x_1,\ldots,x_{i-1},x,x_{i+1},\ldots ,x_k)=1$, we get that $\mathcal{A}'$ outputs an optimal feasible solution of $x'$ with probability at least $p$. Summarizing, if we sample a good $x$ then with probability at least $(0.01-k^{-1/6}) \cdot p\ge 0.005\cdot p=\frac{\ln\left(\nicefrac{1}{1-p}\right)}{c}$ the algorithm $\mathcal{A}'$ in Step 2.5 outputs an optimal feasible solution of $x'$. The proof concludes by noting that a good $x$ is sampled with probability at least $1-k^{-1/6}=1-\frac{\p(d)}{2}>1-\p(d)$.\qedhere\end{sloppypar}
\end{proof}

We conclude this section by providing a couple of remarks on the above theorem and proof. 

\begin{remark}[Amplification factor]\label{rem:soundness}
We would like to note that the amplified hardness for the optimization problem $\Lambda$ (which is shown in Theorem~\ref{thm:main} to be 0.01) can be further amplified to any arbitrarily small positive constant\footnote{Actually, it can be even amplified to sub-constant, but we would have to pay in the running time lower bound for $\Lambda$.} close to 0 by adjusting the parameters in the proof.
\end{remark}

%\begin{remark}[Applicability of the above theorem]\label{rem:appmeta}
%Informally, Theorem~\ref{thm:main} can be used to amplify hardness if the following conditions are met for a $(S,T)$-self direct product feasible problem $\Pi$.
%\begin{itemize}
%\item  $\p(d)=o()$
%\end{itemize}
%\end{remark}

\begin{remark}[Total Problems]\label{rem:total}
One may observe that the class of optimization problems are indeed equivalent to the class of \emph{total} problems. For some finite alphabet $\Sigma$, we call a relation $R\subseteq \Sigma^*\times \Sigma^*$ to be total if for every $x\in\Sigma^*$ there is always a $y\in\Sigma^*$ such that $(x,y)\in R$. To see that every optimization problem $\Pi$ is also a total problem, it suffices to note that the range of the measure function $\Delta_\Pi$ is bounded and therefore a maximum/minimum always exists. And to see that every total problem $R$ is also an optimization problem $\Pi_R$, it suffices to note we can define $\Delta_{\Pi_R}(x,y)$ to be 1 if $(x,y)\in R$ and 0 otherwise, and set $\Pi_R$ to be a maximization problem (i.e., $\goal_{\Pi_R}=\max$).

Given this new outlook at optimization problems, one may see the existence of solutions as a crucial requirement in the proof of Theorem~\ref{thm:main}. In particular, our algorithm $\mathcal{A}$  (design and analysis) would be meaningless without the existence of solutions for any instance of the optimization problem.

\end{remark}

\section{Almost Worst Case to Average Case for Problems in \NP}
\label{sec:NP}

In this section, we show hardness amplification results for various \NP-hard problems, with a focus on \MAXSAT and  Knapsack.

\subsection{Maximum Satisfiability Problem}

We recall below the Maximum satisfiability problem in our formalism for optimization problems.

\begin{definition}[\MAXSAT problem]
	The Maximum Satisfiability problem (\MAXSAT) is an optimization problem characterized by the following quadruple of objects $(I_{\SAT},\SOL_{\SAT},\m_{\SAT},\max)$, where:
	\begin{itemize}
		\item For every $d\in\mathbb{N}$, $I_{\SAT}(d)$ is the set of all \CNF formulas on $n=\Omega(d)$ variables and $O(n)$ clauses; 
		\item For every $\phi\in I_{\SAT}$ we have $\SOL_{\SAT}(\phi)=\{0,1\}^n$;
		\item For every $\phi\in I_{\SAT}$ and every $x\in\{0,1\}^n$  we define $\m_{\SAT}(\phi,x)$ to be the number of clauses that are satisfied by the assignment $x$ to the variables of $\phi$.
	\end{itemize}
\end{definition}

We now show that \MAXSAT is self direct product feasible (in a rather naive way).

\begin{lemma}\label{lem:SAT}
	Let $S,T:\mathbb{N}\times\mathbb{N}\to\mathbb{N}$, where $S(d,k)=dk$ and $T(d,k)=O(dk)$. Then we have that 
	\MAXSAT is $(S,T)$-self direct product feasible.
\end{lemma}
\begin{proof}
 We define the pair of deterministic algorithms $(\G,\Dec)$ below.	
	
For every $\phi_1,\ldots ,\phi_k\in I_{\SAT}(d)$ given as input to $\G$, it outputs the instance $\phi'$ in $I_{\SAT}(d')$ defined as:
$$
\phi' := \phi'_1 \wedge  \dots \wedge \phi'_k
$$
Where $\phi_i'$ is the formula obtained by replacing each literal  $l_j$ in $\phi_i$ by $l_{(i-1)n+j}$. 
It is clear that the running time of $\G$ is $O(d')$. 

Next, for every $i^*\in[k]$, $\phi_1,\ldots ,\phi_k\in I_{\SAT}(d)$, and a \SAT\ assignment $x'\in \zo^{n'}$ for $\phi'$ given as input to $\Dec$, the algorithm first runs $\G$ to compute $\phi'$ and then outputs $x\in \zo^n$ which is computed as follows:
$x_j=x'_{(i^*-1)n+j}$.

We next show that if $x'$ is an optimal \SAT\ assignment for $\phi'$ then $x$ is an optimal assignment for $\phi_{i^*}$. 
This follows easily by the fact that the $\phi'_i$s are defined on disjoint sets of variables, hence an optimal solution for $\phi'$ induces an optimal solution for each of the $\phi_i$s. 

The running times  of $\G$ and $\Dec$ trivially follow.
\end{proof}

The above rather simple self direct product feasibility has a rather strong consequence when combined with our generic hardness amplification theorem.

\begin{corollary}\label{cor:SATpoly}
Let $a\ge 8$.
Let  $\mathcal{D}=\{D_d\}_{d\in\mathbb{N}}$ be a family of distributions such that for every randomized algorithm $\mathcal{A}$ running in time $O(d^a)$ over inputs of size $d$, the following holds for large $d\in\mathbb{N}$.
\begin{itemize}
\item $D_d$ is a distribution over $I_{\SAT}(d)$ where an instance of $I_{\SAT}(d)$ can be sampled from $D_d$ in $\tilde O(d)$ time.
\item $
\underset{\phi\sim D_d}{\Pr}\left[\mathcal A \text{ finds maximizing assignment for }\phi \text{ with probability at least }2/3\right]\le 1-\frac{1}{d}$.
\end{itemize}

Then for $m:=\Theta(d^7)$,  there is a distribution family $\mathcal{D}'=\{D_d'\}_{d\in\mathbb{N}}$ such that for every randomized algorithm $\mathcal{A}'$ running in time $O(m^{a/7})$ over inputs of size $m$, the following holds for all large enough $d\in\mathbb{N}$.
\begin{itemize}
\item $D_d'$ is a distribution over $I_{\SAT}(m)$ and an instance in $I_{\SAT}$ can be sampled from $D_d'$ in $\tilde{O}(m)$ time.
\item $
\underset{\phi'\sim D_d'}{\Pr}\left[\mathcal A' \text{ finds maximizing assignment for }\phi' \text{ with probability at least }2/3\right]\le 0.01.
$
\end{itemize}
\end{corollary} 
\begin{proof}
We apply Theorem~\ref{thm:main} by setting, $\Pi=\Lambda=$ \MAXSAT, $p=\nicefrac{2}{3}$, $S(d,k)=dk$, $T(d,k)=w\cdot dk$ (for some $w\in\mathbb{N}$), $v(d)=w'\cdot d$  (for some $w'\in\mathbb{N}$), $s(d)=\tilde{O}(d)$, $t(d)=d^a$, and $\p(d)=\nicefrac{1}{d}$. Then we have that $k:=\Theta(d^6)$ and $c:=300\ln 3$. We verify that $k\cdot s(d)+T(d,k)+v(d)\le \frac{t(d)}{2c}$ holds by noting that $k\cdot s(d)+T(d,k)+v(d)=\tilde{O}(d^7)$ and $\frac{t(d)}{2c}=\Omega(d^a)=\Omega(d^8)$ (because $a\ge 8$). Therefore, we have that the sampling time from $D'_d$ is $\tilde{O}(d^7)=\tilde{O}(m)$ and that the theorem statement holds for any  randomized algorithm $\mathcal{A}'$ running in time $ \frac{t(d)}{2c}=\Theta(d^a)=\Theta(m^{a/7})$. 
\end{proof}

\begin{remark}\label{rem:maxprobNP}
The idea of taking $k$ instances disjointly as in Lemma~\ref{lem:SAT} can be extended to many \NP-hard \emph{covering} problems such as Vertex Cover, Dominating Set, etc.  Consequently, we obtain hardness amplification results similar to Corollary~\ref{cor:SATpoly} for these problems as well. 
\end{remark}

Now we consider the same hardness amplification result of \MAXSAT but against subexponential time algorithms. We provide below a slightly informal statement (using asymptotic notations as opposed to providing specific constants) for clarity.

\begin{corollary}\label{cor:SATETH}
Let  $\mathcal{D}=\{D_d\}_{d\in\mathbb{N}}$ be a family of distributions such that for every randomized algorithm $\mathcal{A}$ running in time $2^{o(d)}$ over inputs of size $d$, the following holds for large $d\in\mathbb{N}$.
\begin{itemize}
\item $D_d$ is a distribution over $I_{\SAT}(d)$ where an instance of $I_{\SAT}(d)$ can be sampled from $D_d$ in $\tilde O(d)$ time.
\item $
\underset{\phi\sim D_d}{\Pr}\left[\mathcal A \text{ finds maximizing assignment for }\phi \text{ with probability at least }2/3\right]\le 1-\frac{1}{2^{o(d)}}$.
\end{itemize}

Then for $m:=2^{o(d)}$,  there is a distribution family $\mathcal{D}'=\{D_d'\}_{d\in\mathbb{N}}$ such that for every randomized algorithm $\mathcal{A}'$ running in time $m^{\omega(1)}$ over inputs of size $m$, the following holds for all large enough $d\in\mathbb{N}$.
\begin{itemize}
\item $D_d'$ is a distribution over $I_{\SAT}(m)$ and an instance in $I_{\SAT}$ can be sampled from $D_d'$ in $\tilde{O}(m)$ time.
\item $
\underset{\phi'\sim D_d'}{\Pr}\left[\mathcal A' \text{ finds maximizing assignment for }\phi' \text{ with probability at least }2/3\right]\le 0.01.
$
\end{itemize}
\end{corollary}
\begin{proof}Let $h:\mathbb{N}\to\mathbb{N}$ be some slowly increasing function such that $\underset{x\rightarrow \infty}{\lim}\ h(x)=\infty$ and let $h:=h(d)$. 
We apply Theorem~\ref{thm:main} by setting, $\Pi=\Lambda=$ \MAXSAT, $p=\nicefrac{2}{3}$, $S(d,k)=dk$, $T(d,k)=w\cdot dk$ (for some $w\in\mathbb{N}$), $v(d)=w'\cdot d$  (for some $w'\in\mathbb{N}$), $s(d)=\tilde{O}(d)$, $t(d)=2^{d/h}$, and $\p(d)=\nicefrac{1}{2^{d/(6\cdot h^2)}}$. Then we have that $k:=\Theta(2^{d/h^2})$ and $c:=300\ln 3$. We verify that $k\cdot s(d)+T(d,k)+v(d)\le \frac{t(d)}{2c}$ holds by noting that $k\cdot s(d)+T(d,k)+v(d)=2^{(1+o(1))\cdot d/h^2}$ and $\frac{t(d)}{2c}=\Omega(2^{d/h})$. Therefore, we have that the sampling time from $D'_d$ is $2^{(1+o(1))\cdot d/h^2}=O(m)$ and that the theorem statement holds for any  randomized algorithm $\mathcal{A}'$ running in time $ \frac{t(d)}{2c}=O(2^{d/h})=\Theta(m^{h(d)})=m^{\omega(1)}$. 
\end{proof}

Note that if we assume the (randomized) Exponential Time Hypothesis (\ETH{}) for 3-\SAT\ \cite{IP01,IPZ01,CIP06} then after applying the Sparsification lemma, we obtain for some $\varepsilon>0$, a family of distributions $\mathcal{D}=\{D_d\}_{d\in\mathbb{N}}$  such that  for every randomized algorithm $\mathcal{A}$ running in time $2^{\varepsilon d}$ over inputs of size $d$, the following holds for large $d\in\mathbb{N}$.
\begin{itemize}
\item $D_d$ is a distribution over $I_{\SAT}(d)$ where an instance of $I_{\SAT}(d)$ can be sampled from $D_d$ in $\tilde O(d)$ time.
\item $
\underset{\phi\sim D_d}{\Pr}\left[\mathcal A \text{ finds maximizing assignment for }\phi \text{ with probability at least }2/3\right]\le 1-\frac{1}{2^{\tilde{O}(d)}}$.
\end{itemize} 

Therefore, our amplification in Corollary~\ref{cor:SATETH} is an \emph{almost} worst-case to average-case reduction for \MAXSAT (under subexponential time reductions).

\subsection{Knapsack Problem}

In this subsection, we study the direct product feasibility of the Knapsack problem and as we will see, showing that its direct product feasible for reasonable parameters is significantly more non-trivial than was with the case for \MAXSAT.

In the Knapsack problem we are given a target \emph{sack weight} $W$, and set of items via pairs $(w,v)$ where $w$ is the weight of the item and $v$ is the value of the item, and the goal is to pick a subset of items which maximizes the sum of the values of the picked items  given the constraint that  their total weight  is at most $W$.  More formally, we describe it as follows.

\begin{definition}[$\KS$ problem]
	The Knapsack problem ($\KS$) is an optimization problem characterized by the following quadruple of objects $(I_{\KS},\SOL_{\KS},\m_{\KS},\max)$, where:
	\begin{itemize}
		\item For every $d\in\mathbb{N}$, $I_{\KS}(d)=(W, \set{(v_i,w_i)}_{i=1}^n)$  where $W,n, v_i,w_i\in \N$ such that $d= \log W+\sum_{i=1}^n(\log v_i+\log w_i )$; 
		\item For every $(W, \set{(v_i,w_i)}_{i=1}^n)\in I_{\KS}$ we have $\SOL_{\KS}((W, \set{(v_i,w_i)}_{i=1}^n))$ is the set of all subsets $S$ of $[n]$ satisfying $\sum _{i\in S}w_i\le W$;
		\item For every $S\in 2^{[n]}$  satisfying $\sum _{i\in S}w_i\le W$ we define $\m_{\KS}(S)$ to be $\sum _{i\in S}v_i$.
	\end{itemize}
\end{definition}

It is not trivial to show the self direct problem feasibility for Knapsack. 
To see that, consider the case $k=2$. Naively, if the sacks weights are $W_1, W_2$ 
then one may define a new sack of weight $W_1+W_2$, and then take the union of the item sets while leaving their weights and values untouched. 
However, in this simple reduction, we may use some of the target sack weight of one instance 
against another instance. Nonetheless we show with some care, a direct product feasibility result can be  obtained. 

\begin{lemma}\label{lem:sdfKS}
Let $\mathbb T:=\{2^\ell\mid \ell\in\mathbb{N}\}$.	Let $S,T:\mathbb{N}\times\mathbb{T}\to\mathbb{N}$, where $S(d,k)=k^{O(1)}\cdot d$ and $T(d,k)=k^{O(1)}\cdot d$. Then we have that 
	$\KS$ is $(S,T)$-self direct product feasible.
\end{lemma}

\begin{proof}
	We first show that there exists a pair of deterministic  algorithms $(\G,\Dec)$ such that the conditions of $(S,T)$-self direct product feasibility in Definition~\ref{def:DPfeasible} are met for the Knapsack problem for every $d\in\mathbb{N}$ but when $k$ is fixed to be $2$.  Using that, we prove the lemma statement 
	for any value of $k$ which is a power of $2$. 
	The proof proceeds by first creating (recursively) two instances:
	the first corresponds to the first $k/2$ instances 
	and the second to the last $k/2$ ones. 
	Then we use the result for $k=2$ to create a single instance.

\paragraph{Base Case $k=2$.}
	Let us first present the algorithm $ \G$.
	Let $I_1, I_2 \in I_\KS(d)$ be the input to $\G$ where for every $j\in\{1,2\}$, we have $I_j=(W_j,\set{(v^{(j)}_i, w^{(j)}_i)}_{i\in [n_j]})$. We first normalize the weights (by multiplying the weight of both the sack and 
	each item by the same factor $c$) so that $W_1=W/2, W_2=W$, for some $W\in\mathbb{N}$ (note that we can achieve this normalization for $W= 2\cdot W_1\cdot W_2$).
	
	The output of $\G$ is a new instance $I':=(W',\set{(v^{'}_i, w^{'}_i)}_{i\in [n']})\in I_{\KS}(d')$, where $d'=O(d^2)$, $n':=n_1+n_2+\log W$, and   $W'=W^2+W/2$. We define $N_1=\set{1,\dots, n_1}, 
	N_2= \set{n_1+1, \dots, n_2}$, and $D=\set{n_1+n_2+1,\dots, n_1+n_2+\log W}$.
		
		Now we define the items $(v^{'}_i, w^{'}_i)$ for all $i\in[n']$.
		The first $n_1$ items correspond to the items of $I_1$, the next two $n_2$ items 
		correspond to the items in $I_2$, and the last $\log W$ are dummy items. Elaborating, we have 
\begin{align*}
\forall i\in [n'],\ v_i'&=\begin{cases}
v_i^{(1)}\text{ if }i\le n_1\\
v^{(2)}_{i-n_1}\cdot (m+1) W+1\text{ if }n_1<i\le n_1+n_2\\
(m+1)\cdot 2^{i-n_1-n_2-1}\text{ if }i>n_1+n_2
\end{cases},
\\
\forall i\in [n'],\ w_i'&=\begin{cases}
w_i^{(1)}\text{ if }i\le n_1\\
w^{(2)}_{i-n_1}\cdot W\text{ if }n_1<i\le n_1+n_2\\
W\cdot 2^{i-n_1-n_2-1}\text{ if }i>n_1+n_2
\end{cases},
\end{align*}		
		
		where $m:=\underset{i\in [n_1]}{\sum}v^{(1)}_{i}$. Note that the size of $I'$ is indeed $O(d)$.

		Now we define  $\Dec$: It gets as  input an index $j\in \{1,2\}$, 
		an instance $I'$ which was generated by $\G$ and an optimal solution $S'$ for $I'$. 
		It is required to produce an optimal solution for the instance $I_j$. 
		The algorithm $\Dec$ returns $S'\cap N_1$ if $j=1$, otherwise it outputs $S'\cap N_2$ 
		(where each element is translated by $-n_1$ so that the final output resides in $[n_2]$).
		
		The correctness of $\Dec$  follows by the following claim.
		
		\begin{claim}
			Let $S'$ be an optimal solution for $I'$, then:
			\begin{enumerate}
				\item $S'_1:= S'\cap N_1$ is an optimal solution for $I_1$.
				\item Let $S_2:= S'\cap N_2$ and define $S'_2$ as the set of elements in $S_2$ 
				where each element is translated by $-n_1$, then  $S'_2$ is an optimal solution for $I_2$.
			\end{enumerate}
		\end{claim} 
	
		\begin{proof}
			
		\begin{enumerate}
			
			\item We first show that for $S'$ we have 
			$${\bf w}:=\underset{\sett {i'\in S'} {i'\in N_2 \cup D}}{\mathlarger{\sum}}w'_i=W^2.$$ 
			Assume not, then either we have ${\bf w}>W^2$ or ${\bf w}<W^2$. If it is the former then notice that since $w_i'$ is a multiple of $W$ for all $i>n_1$, we arrive at a contradiction as ${\bf w}\le W'=W^2+W/2$. Therefore let us assume that it is the latter.  We define a `better' solution $S''$ as follows. 
			
			First include into $S''$ all elements in $S'\cap N_2$. 
			Let $\rho=W^2-\sum_{i' \in S' \cap N_2}w'_i$, denote the remaining slack in the sack after inserting 
			the elements in $N_2 \cap S'$. Next, for each $i\in D$ we insert $i$ into $S''$ if in the binary 
			representation of $\rho$, the $i$-th bit equals $1$. We show that  $\m(S'')> \m(S')$, 
			contradicting the optimality of $S'$.
			
			Let $D'=S'\cap D$, and let $\rho'=\sum_{i\in D'}2^{i-(n_1+n_2)-1}$ be the number obtained by the binary 
			representation of the elements in $D'$. 
			Observe that since by our assumption ${\mathlarger{\sum}}_{\sett {i'\in S'} {i'\in N_2 \cup D}}w'_i<W^2$ we get  $\rho\ge \rho'+1$ and hence:   
			$$\m (\card{S'' \cap D})-  \m (\card{S' \cap D})=(\rho-\rho')\cdot (m+1) \ge m+1, $$ and we conclude:
			\begin{eqnarray*}
				\m(S'')-\m(S')&=&  	\m (\card{S'' \cap D})-  \m (\card{S' \cap D}) + \m (\card{S'' \cap N_1})-  \m (\card{S' \cap N_1}) \\
							  &\ge & m+1  - \sum_{i\in S'\cap N_1}v_i{'} \\
							  &\ge& 1,
			\end{eqnarray*}
			
			where the last inequality follows since $\sum_{i\in S'\cap N_1}v_i'\le m=\underset{i\in [n_1]}{\sum}v^{(1)}_{i}$, contradicting the optimality of $S'$.

			Now, clearly, if ${\mathlarger{\sum}}_{\sett {i'\in S'} {i'\in N_2 \cup D}}w_i=W^2$, then $S'\cap N_1$ is an optimal solution for $I_1$, since otherwise we can improve over the solution $S'$ (by taking the same items from $N_2$ and $D$ and add the optimal solution for $I_1$).
			
			\item Assume for sake of contradiction that, $S'_2$ defined in the claim, is not an optimal solution for $I_2$. 
			Let $\tilde {S_2}$ be an optimal solution for $I_2$.
			Let us define $S''$ as follows: First we include the items from $\tilde S_2$, then add items in $D$ until 
			the total weight reaches $W^2$. 
			Finally, we include the items from $S'\cap N_1$. 
			
			Observe that by the previous item, the set $S''$ is a feasible solution 
			(as the weight of the elements in $S'\cap N_1$ does not exceeds $W/2$). 
%			Furthermore, by the optimality of $\tilde{S}$ for instance $I_2$  we have:
%			$\m(\tilde{S_2})\ge \m (S'_2)+1$, 
			Since in $I'$ for each $i\in N_2$ we set: $v'_{i}=v^{(2)}_{i-n_1}\cdot (m+1) W+1$ and by the optimality of $\tilde {S}_2$ we have:
			$$ \m (S''\cap N_2)- \m (S'\cap N_2)\ge (m+1) W+1. $$
			
			Observe that $S'$ may contain at most $\log W$ more elements from $D$ than $S''$ contains. However their value is bounded by $(m+1)W$, and hence:
			\begin{eqnarray*}
				\m(S'')-\m(S')&=& (\m(S''\cap N_2)-\m(S'\cap N_2))+(\m(S''\cap D)-\m(S'\cap D))\\
				&\ge& (m+1) W+1 -(m+1) W\\
				& \ge& 1.
			\end{eqnarray*}
			Contradicting the optimality of $S'$.\qedhere
		\end{enumerate}		
		\end{proof}
		Finally, it is easy to see that running time of $\Dec$ is at most $O(d')$. 

		\paragraph{General Case $k=2^\ell$, for some $\ell\in\mathbb{N}$.} We will use $(\G,\Dec)$ given in the previous case to  show that there exists a pair of deterministic  algorithms $(\tilde \G,\tilde \Dec)$ such that the conditions of $(S,T)$-self direct product feasibility in Definition~\ref{def:DPfeasible} are met for the Knapsack problem for every $d\in\mathbb{N},k\in\mathbb{T}$. 

Let us first present the algorithm $ \tilde\G$.
	Let $I_1, \ldots I_k \in I_\KS(d)$ be the input to $\G$ where for every $j\in[k]$, we have $I_j=(W_j,\set{(v^{(j)}_i, w^{(j)}_i)}_{i\in [n_j]})$. We arbitrarily pair up the $k$ instances, and feed each pair of instances  to $\G$ (described previously in the proof). We obtain $k/2$ instances of Knapsack problem of size $O(d)$. We repeat the process of arbitrarily pairing up the instances and feeding it to $\G$. After doing this process $\log k$ times, we will have a single instance of Knapsack as the output of $\G$ which will be of size at most $k^{O(1)}\cdot d$. This is the  output of $\tilde \G$.

Finally, it suffices to note that $\tilde \Dec$ simply does a restriction of the solution to the coordinates of the instance of interest as in the previous case where $k$ was set to 2. 
\end{proof}

Like \MAXSAT, Knapsack too admits a hardness amplification result but we skip writing it here for the sake of non-repetitiveness. 
Also, notice that the above lemma is proven for functions $S,T$ on domain $\mathbb{N}\times
\mathbb{T}$ instead of $\mathbb{N}\times
\mathbb{N}$. This is done for the sake of clear presentation. The above proof can be extended to prove the direct product feasibility for the general case as well.

\begin{remark}[Adopting above proof to maximization version of other covering problems]\label{rem:othermax}
The idea of taking $k$ instances with appropriate scaling as in Lemma~\ref{lem:sdfKS} can be extended to many \emph{maximization} versions of {covering} problems (that are \NP-hard) such as Max-coverage, Clustering etc.  Consequently, we obtain hardness amplification results  for these problems as well. 
\end{remark}

\section{Mild Average Case to Sharp Average Case for Problems in \P}\label{sec:P}
In this section, we look at hardness amplification for two natural and important string problems which have been at the center stage of fine-grained complexity in the last few years. We also look at the problem of matrix multiplication, which does \emph{not} fit into our scheme of hardness amplification given in Section~\ref{sec:main} as it is not known to admit efficient deterministic verification. We also propose the problem of computing Frechet distance as a natural problem which might not be direct product feasible. 
 
\subsection{Longest Common Subsequence}

In the Longest Common Subsequence problem we are given two strings and the goal is to find the subsequence of maximum length that is in both the strings. It is indeed a natural maximization problem and fits smoothly into our formalism as follows.

\begin{definition}[\LCS alignment]
Let $\Sigma$ be a finite non-empty set and $n\in\mathbb{N}$. For every pair of strings $(a,b)\in \Sigma^n\times \Sigma^n$ and every function $\sigma:[n]\to[n]\cup\{\perp\}$, we say that $\sigma$ is an \LCS alignment for $(a,b)$ if the following holds.
\begin{itemize}
\item \textbf{Monotonicity}: For every $i,j\in [n]$, $i<j$ we have that if $\sigma(i)\neq \perp$ and $\sigma(j)\neq \perp$ then $\sigma(i)<\sigma(j)$.
\item \textbf{Matching}:  For every $i\in [n]$,  if $\sigma(i)\neq \perp$ then we have $a_i=b_{\sigma(i)}$.
\end{itemize}
The length of an \LCS alignment $\sigma$ is defined as the cardinality of the preimage of $[n]$, i.e., $$\left|\{i\in[n]\mid \sigma(i)\in[n]\}\right|.$$
\end{definition}

\begin{definition}[\LCS problem]
Let $\Sigma$ be a finite non-empty set. The Longest Common Subsequence ($\LCSsig$) is an optimization problem charaterized by the following quadruple of objects $(I_{\LCSsig},\SOL_{\LCSsig},\m_{\LCSsig},\max)$, where:
\begin{itemize}
\item For every $d\in\mathbb{N}$, $I_{\LCSsig}(d)=\Sigma^{n}\times \Sigma^{n}$ where\footnote{Here we use the unary encoding to encode a symbol in $\Sigma$ for the ease of presentation, as it circumvents rounding issues. This does not affect the results in this paper as we think of $\Sigma$ as some small universal constant (like $\Sigma=\{0,1\}$). The results in this paper also hold for larger alphabets, but more care needs to taken in the rounding of parameters in the forthcoming proofs while using the binary encoding.} $n=d/(2|\Sigma|)$; 
\item For every $(a,b)\in I_{\LCSsig}$ we have $\SOL_{\LCSsig}(a,b)$ is the set of all \LCS alignments for $(a,b)$;
\item For every $(a,b)\in I_{\LCSsig}$ and every \LCS alignment $\sigma$ for $(a,b)$ we define $\m_{\LCSsig}(a,b,\sigma)$ to be the length of $\sigma$.
\end{itemize}
\end{definition}

Now we show that \LCS on alphabets of some size are direct product feasible with \LCS on alphabets with an additional character.

\begin{lemma}\label{lem:DPLCS}
Let $\Sigma$ be a finite non-empty set. Let $\Xi$ be a superset of $\Sigma$ of cardinality $|\Sigma|+1$. Let $S,T:\mathbb{N}\times\mathbb{N}\to\mathbb{N}$, where $S(d,k)=2|\Xi|(dk^2+k-1)$ and $T(d,k)=2\cdot S(d,k)$. Then we have that 
$(\LCSsig,\LCSxi)$ are $(S,T)$-direct product feasible.
\end{lemma}
\begin{proof}
Let $\Xi=\Sigma\cup\{\xi\}$, where $\xi\notin\Sigma$. We define the pair of deterministic algorithms $(\G,\Dec)$ below.

Fix $k,d\in \mathbb N$ (and consequently $n\in\mathbb{N}$). Let $\ell:=1+nk$, $m:=nk+\ell(k-1)$, and $d'=2m|\Xi|$. 
Let $\bxi\in\Xi^\ell$ be the concatenation of $\ell$ copies of $\xi$, i.e., 
\begin{align}
\bxi:=\underbrace{\xi\circ\cdots \circ \xi}_{\ell \text{ copies}}.\nonumber
\end{align}

For every $(a_1,b_1),\ldots ,(a_k,b_k)\in I_{\LCSsig}(d)$ given as input to $\G$, it outputs the  instance $(\mathbf{a,b})$ in $I_{\LCSxi}(d')$ where, $$\mathbf{a}:=a_1\circ\bxi\circ a_2\circ \bxi\circ\cdots \circ \bxi\circ a_k\in \Xi^m,\ \ \mathbf{b}:=b_1\circ \bxi\circ b_2\circ \bxi\circ\cdots \circ \bxi\circ b_k\in \Xi^m.$$
It is clear that the running time of $\G$ is $d'$. 

Next, for every $i\in[k]$, $(a_1,b_1),\ldots ,(a_k,b_k)\in I_{\LCSsig}(d)$, and an \LCS alignment $\boldsymbol{\tilde \sigma}:[m]\to[m]\cup\{\perp\}$ for $(\mathbf{a,b})$ given as input to $\Dec$, the algorithm first runs $\G$ to compute $(\mathbf{a,b})$ and then outputs $\sigma:[n]\to[n]\cup\{\perp\}$ which is computed as follows:
$$
\forall j\in[n],\ \sigma(j)=\begin{cases}
\boldsymbol{\tilde\sigma}(j+\texttt{loc}) \text{ if }\texttt{loc}< \boldsymbol{\tilde\sigma}(j+\texttt{loc})\le n+\texttt{loc},\\
\perp\text{ otherwise, }
\end{cases}
$$
where $\texttt{loc}=(n+\ell)(i-1)$. It is easy to see that $\sigma$ is an \LCS alignment for $(a_i,b_i)$. Also, it is easy to see that the running time of $\Dec$ is running time of $\G$ plus $d'$ (needed to compute $\sigma$). 

We next show that if $\boldsymbol{\tilde\sigma}$ is an optimal \LCS alignment for $(\mathbf{a,b})$ then $\sigma$ is an optimal alignment for $(a,b)$. To show this we first  show that if $\boldsymbol{\tilde\sigma}$ is an optimal \LCS alignment  then it has a certain ``block structure''. Consider the following \LCS alignment $\boldsymbol{\tilde\tau}$ for $(\mathbf{a,b})$:
$$
\forall j\in[m],\ \boldsymbol{\tilde\tau}(j)=\begin{cases}
j \text{ if }\mathbf{a}_j=\xi,\\
\perp\text{ otherwise. }
\end{cases}
$$
$\boldsymbol{\tilde\tau}$ is an \LCS alignment because of our construction of $(\mathbf{a,b})$, where we have that for all $j\in[m]$ if $\mathbf{a}_j=\xi$ then $\mathbf{b}_j=\xi$ as well. We have that $\m_{\LCSxi}(\boldsymbol{a,b,\tilde\tau})=\ell(k-1)$. Therefore if $\boldsymbol{\tilde\sigma}$ is an optimal \LCS alignment then we should have $\m_{\LCSxi}(\boldsymbol{a,b,\tilde\sigma})\ge\ell(k-1)$. 

Suppose there exists $j\in[n]$ such that $\boldsymbol{\tilde\sigma}(j+\texttt{loc})\neq \perp $ and is strictly greater than $n+\texttt{loc}$. If there is more than one such $j$ then we pick the smallest one.  Then we have that $\boldsymbol{\tilde\sigma}(j+\texttt{loc})>n+\texttt{loc}+\ell$ as $\mathbf{a}_{j+\texttt{loc}}\in\Sigma$ but $\mathbf{b}_{n+\texttt{loc}+r}=\xi\notin\Sigma$ for all $r\in[\ell]$. Also we have that for every $j'\in [n]$ that is strictly less than $j$ either $\boldsymbol{\tilde\sigma}(j+\texttt{loc})= \perp $ or is at most $n+\texttt{loc}$. In this case, we have that the length of $\boldsymbol{\tilde\sigma}$ is at most $m-\ell=nk+\ell(k-2)=\ell(k-1)-1$, a contradiction as we showed earlier that $\m_{\LCSxi}(\boldsymbol{a,b,\tilde\sigma})\ge\ell(k-1)$. By following a similar argument we can show that there does not exist $j\in[n]$ such that $\boldsymbol{\tilde\sigma}(j+\texttt{loc})\neq \perp $ and is less than or equal to $\texttt{loc}$. This implies that $\boldsymbol{\tilde\sigma}$ restricted to the coordinates in the interval $[\texttt{loc}+1,\texttt{loc}+n]$ provides an optimal \LCS alignment for pair of contiguous substring of $\mathbf{a}$ and $\mathbf{b}$ restricted to the coordinates in the interval $[\texttt{loc}+1,\texttt{loc}+n]$. Thus $\sigma$ is an optimal \LCS alignment for $(a_i,b_i)$.

Finally, we note that the bound on the functions $S$ and $T$ hold as $d'=(nk+(nk+1)(k-1))2|\Xi|=2|\Xi|(nk^2+k-1)\le 2|\Xi|(dk^2+k-1) $.
\end{proof}

\LCS problem has been of special interest in the last few years thanks to the advancements in fine grained complexity \cite{ABW15,AHWW16,Williams15,Williams16,Vir18}. The above direct product feasibility immediately implies the below hardness amplification result.

\begin{corollary}\label{cor:LCS}
Let $\varepsilon>0$. Let $\Sigma$ be a finite non-empty set. Let $\Xi$ be a superset of $\Sigma$ of cardinality $|\Sigma|+1$. 
Let  $\mathcal{D}=\{D_d\}_{d\in\mathbb{N}}$ be a family of distributions such that for every randomized algorithm $\mathcal{A}$ running in time $d^{1+\varepsilon}$ over inputs of size $d$, the following holds for large $d\in\mathbb{N}$.
\begin{itemize}
\item $D_d$ is a distribution over $I_{\LCSsig}(d)$ where an instance of $I_{\LCSsig}(d)$ can be sampled from $D_d$ in $\tilde{O}(d)$ time.
\item $
\underset{(a,b)\sim D_d}{\Pr}\left[\mathcal A \text{ finds optimal \LCS alignment for }(a,b) \text{ with probability at least }2/3\right]\le 1-d^{-o(1)}$.
\end{itemize}

Then there is some $\varepsilon'>0$ such that there is a distribution family $\mathcal{D}'=\{D_d'\}_{d\in\mathbb{N}}$ such that for every randomized algorithm $\mathcal{A}'$ running in time $d^{1+\varepsilon'}$, the following holds for all large enough $d\in\mathbb{N}$.
\begin{itemize}
\item $D_d'$ is a distribution over $I_{\LCSxi}(d^{1+o(1)})$ and an instance in $I_{\LCSxi}$ can be sampled from $D_d'$ in $d^{1+o(1)}$ time.
\item $
\underset{(a',b')\sim D_d'}{\Pr}\left[\mathcal A' \text{ finds optimal \LCS alignment for }(a',b') \text{ with probability at least }2/3\right]\le 0.01.
$
\end{itemize}
\end{corollary}
\begin{proof}
We apply Theorem~\ref{thm:main} by setting, $\Pi=I_{\LCSsig}$, $\Lambda=I_{\LCSxi}$, $p=\nicefrac{2}{3}$, $S(d,k)=2|\Xi|(dk^2+k-1)$, $T(d,k)=2\cdot S(d,k)$, $v(d)=w\cdot d$  (for some $w\in\mathbb{N}$), $s(d)=\tilde{O}(d)$, $t(d)=d^{1+\varepsilon}$, and $\p(d)=\nicefrac{1}{d^{o(1)}}$. Then we have that $k:=d^{o(1)}$ and $c:=300\ln 3$. We verify that $k\cdot s(d)+T(d,k)+v(d)\le \frac{t(d)}{2c}$ holds by noting that $k\cdot s(d)+T(d,k)+v(d)=d^{1+o(1)}$ and $\frac{t(d)}{2c}=\Omega(d^{1+\varepsilon})$. Therefore, we have that the sampling time from $D'_d$ is $d^{1+o(1)}$ and that the theorem statement holds for any  randomized algorithm $\mathcal{A}'$ running in time $ \frac{t(d)}{2c}=\Theta(d^{1+\varepsilon})=\Theta\left(\left(d^{1+o(1)}\right)^{1+\varepsilon-\delta}\right)$, for any $\delta>0$. 
\end{proof}

\begin{remark}[Hardness Amplification of $k$-\LCS and other parameterized complexity problems]\label{rem:param}
We remark here that the above proof strategy can be extended to show a hardness amplification result for the $k$-\LCS problem (for fixed $k$) which is of interest in parameterized complexity. In fact, following Remark~\ref{rem:maxprobNP}, we can obtain hardness amplification theorems for fundamental problems  in fixed parameter complexity such $k$-clique and $k$-set cover.
\end{remark}

\subsection{Edit Distance}
Recall that the edit distance between a pair of strings $a,b\in \Sigma^n$ is defined as the minimal number of edit operations needed to convert $x$ into $y$, where edit operations are character insertions/ deletions and substitutions.

\begin{definition}[\ED alignment]
	Let $\Sigma$ be a finite non-empty set and $n\in\mathbb{N}$. For every pair of strings $(a,b)\in \Sigma^n\times \Sigma^n$ and every function $\sigma:[n]\to[n]\cup\{\perp\}$, we say that $\sigma$ is an \LCS alignment for $(a,b)$ if the following holds.

 \textbf{Monotonicity}: For every $i,j\in [n]$, $i<j$ we have that if $\sigma(i)\neq \perp$ and $\sigma(j)\neq \perp$ then $\sigma(i)<\sigma(j)$.

	The cost of an \ED alignment $\sigma$ is defined as twice the cardinality of the preimage of $\perp$ plus the number of mismatches, i.e., $$2\left|\{i\in[n]\mid \sigma(i)=\perp \}\right|+\left|\{i\in[n]\mid y_\sigma(i)\neq x_i \}\right| .$$
\end{definition}
We interpret an alignment $\sigma$ is as follows: For every $i\in [n]$ if $\sigma(i)\neq \perp$ then $\sigma$ matches each symbol $x_i$ into $y_{\sigma(i)}$ (while paying an edit operation in case of substitution). In case $\sigma(i)=\perp$, then it means that $\sigma$ deletes the $i$-th character of $x$. In case we have $\sigma (i), \sigma (i+1) \neq \perp$ but $\sigma (i+1)>\sigma(i)+1$ then it means $\sigma$ inserts the characters $y_{\sigma(i)+1},\dots, y_{\sigma(i+1)-1}$. The bound on the edit cost of $\sigma$ simply follows by the fact that the number of characters insertions and deletions is equal.
\begin{definition}[\ED problem]
	Let $\Sigma$ be a finite non-empty set. The Edit Distance ($\EDsig$) is an optimization problem characterized by the following quadruple of objects $(I_{\EDsig},\SOL_{\EDsig},\m_{\EDsig},\min)$, where:
	\begin{itemize}
		\item For every $d\in\mathbb{N}$, $I_{\EDsig}(d)=\Sigma^{n}\times \Sigma^{n}$ where\footnote{Here again we use the unary encoding to encode a symbol in $\Sigma$ for the ease of presentation.} $n=d/(2|\Sigma|)$; 
		\item For every $(a,b)\in I_{\EDsig}$ we have $\SOL_{\EDsig}(a,b)$ is the set of all \ED alignments for $(a,b)$;
		\item For every $(a,b)\in I_{\EDsig}$ and every \ED alignment $\sigma$ for $(a,b)$ we define $\m_{\EDsig}(a,b,\sigma)$ to be the cost of $\sigma$.
	\end{itemize}
\end{definition}

\begin{lemma}\label{lem:DPEdit}
	Let $\Sigma$ be a finite non-empty set. Let $\Xi$ be a superset of $\Sigma$ of cardinality $|\Sigma|+1$. Let $S,T:\mathbb{N}\times\mathbb{N}\to\mathbb{N}$, where $S(d,k)=3|\Xi|dk^2$ and $T(d,k)=2\cdot S(d,k)$. Then we have that 
	$(\EDsig,\EDxi)$ are $(S,T)$-direct product feasible.
\end{lemma}
\begin{proof}
	Let $\Xi=\Sigma\cup\{\xi\}$, where $\xi\notin\Sigma$. We define the pair of deterministic algorithms $(\G,\Dec)$ below.
	
	Fix $k,d\in \mathbb N$ (and consequently $n\in\mathbb{N}$). Let $\ell:=1+nk$, $m:=(n+\ell) k$, and $d'=2m|\Xi|$. 
	Let $\bxi\in\Xi^\ell$ be the concatenation of $\ell$ copies of $\xi$, i.e., 
	\begin{align}
	\bxi:=\underbrace{\xi\circ\cdots \circ \xi}_{\ell \text{ copies}}.\nonumber
	\end{align}
	
	For every $(a_1,b_1),\ldots ,(a_k,b_k)\in I_{\EDsig}(d)$ given as input to $\G$, it outputs the  instance $(\mathbf{a,b})$ in $I_{\EDxi}(d')$ where, $$\mathbf{a}:=a_1\circ\bxi\circ a_2\circ \bxi\circ\cdots \circ \bxi\circ a_k\circ \bxi \in \Xi^m,\ \ \mathbf{b}:=b_1\circ \bxi\circ b_2\circ \bxi\circ\cdots \circ \bxi\circ b_k \circ \bxi\in \Xi^m.$$
	It is clear that the running time of $\G$ is $d'$. 
	
	Next, for every $i\in[k]$, $(a_1,b_1),\ldots ,(a_k,b_k)\in I_{\EDsig}(d)$, and an \ED alignment $\boldsymbol{\tilde \sigma}:[m]\to[m]\cup\{\perp\}$ for $(\mathbf{a,b})$ given as input to $\Dec$, the algorithm first runs $\G$ to compute $(\mathbf{a,b})$ and then outputs $\sigma:[n]\to[n]\cup\{\perp\}$ which is computed as follows:
	$$
	\forall j\in[n],\ \sigma(j)=\begin{cases}
	\boldsymbol{\tilde\sigma}(j+\texttt{loc}) \text{ if }\texttt{loc}< \boldsymbol{\tilde\sigma}(j+\texttt{loc})\le n+\texttt{loc},\\
	\perp\text{ otherwise, }
	\end{cases}
	$$
	where $\texttt{loc}=(n+\ell)(i-1)$. It is easy to see that $\sigma$ is an \ED alignment for $(a_i,b_i)$. Also, it is easy to see that the running time of $\Dec$ is running time of $\G$ plus $d'$ (needed to compute $\sigma$). 
	
	We next show that if $\boldsymbol{\tilde\sigma}$ is an optimal \ED alignment for $(\mathbf{a,b})$ then $\sigma$ is an optimal alignment for $(a,b)$. 
	
	To show this we first prove that there exists an optimal alignment $\boldsymbol{\tilde\sigma}'$ which is an optimal \ED alignment for $(\mathbf{a,b})$ and it is `block-consistent'. Similarly to the \LCS case, every optimal alignment must satisfy: $\boldsymbol{\tilde\sigma}(i)\in \{i-\ell/2,\dots, i+\ell /2\}$ (as otherwise its cost exceeds the identity alignment cost). To simplify the notations we define for each $j\in [k]$ the starting and the end indices of each block $a_i$, specifically: $s(a_j)=(\ell -1 + n)(j-1)+1, f(a_j)=(\ell -1 + n)(j-1) +n$.
	
	We denote by $\Delta (\boldsymbol{\tilde\sigma})_j$ the number of edit operations made by $\boldsymbol{\tilde\sigma}$ on the $j^{\text{th}}$ block of $(\mathbf{a,b})$, i.e., it equals
	twice the cardinality of the preimage of $\perp$ plus the number of mismatches for indices in $j\in \{s(a_j), \dots s(a_j)+n+\ell \}$
	\begin{claim}\label{obs:blockConsistent}
		There exists an optimal alignment $\boldsymbol{\tilde\sigma}'$ which is an optimal \ED alignment for $(\mathbf{a,b})$ which is `block-consistent', that is:
		For all $j\in [k]$, if $j>1$ then,  $$\boldsymbol{\tilde\sigma}'(s(a_j)-1)= s(a_j)-1 \text{ and }\ \boldsymbol{\tilde\sigma}'(f(a_j)+1)= f(a_j)+1.$$
		 
		 Moreover, $$\Delta (\boldsymbol{\tilde\sigma}')_j=\Delta_e(a_j,b_j).$$
	\end{claim}
	\begin{proof}[Proof of Claim~\ref{obs:blockConsistent}]
	We take any optimal alignment $\boldsymbol{\tilde\sigma}$ and gradually change it to be 'block-consistent' while preserving its cost. This is  done as follows:
	We define $\boldsymbol{\tilde\sigma}^0=\boldsymbol{\tilde\sigma}$. For each $j\in [k]$ we define $\boldsymbol{\tilde\sigma}^j=\boldsymbol{\tilde\sigma}^{j-1}$ and then convert it to be consistent on the $j^{\text{th}}$ block by: 
	\begin{itemize}
		\item Deleting all the characters in $a_j$ that were mapped into $\bxi$ and then matching all the characters of $\bxi$
		(formally, for each $i\in \set{s(a_j), \dots, f(a_j)}$ if $\boldsymbol{\tilde\sigma}^{j-1}(i)>f(a_j)$ set  $\boldsymbol{\tilde\sigma}^{j}(i)=\perp$); 
		\item Matching all the characters in $j^{\text{th}}$ block of $\bxi$ (formally, for each $i\in \set{f(a_j)+1, \dots, s(a_{j+1})-1}$ set $\boldsymbol{\tilde\sigma}^{j}(i)=i$);
		\item Each character in the prefix of $a_{j+1}$ that was mapped into $\bxi$ in $\boldsymbol{\tilde\sigma}^{j-1}$ is deleted (formally for for each $i\in \set{s(a_{j+1}), \dots, f(a_{j+1})}$ if $\boldsymbol{\tilde\sigma}^{j-1}(i)<s(a_{j+1})$ set  $\boldsymbol{\tilde\sigma}^{j}(i)=\perp$).
	\end{itemize}
	Finally set $\boldsymbol{\tilde\sigma}'=\boldsymbol{\tilde\sigma}^k$.
	
	We next claim that $\m (\boldsymbol{\tilde\sigma}^{j})\le \m (\boldsymbol{\tilde\sigma}^{j-1})$ (and by the optimality of $\boldsymbol{\tilde\sigma}^0$ we get an optimality of $\boldsymbol{\tilde\sigma}^j$): Let $\disp_{\boldsymbol{\tilde\sigma}^{j}}(i)=i-\boldsymbol{\tilde\sigma}^{j}(i)$. The proof proceeds by a case analysis:
	
	\paragraph{Case 1 -- $\disp_{\boldsymbol{\tilde\sigma}^{j-1}}(f(a_j)+1)> 0$:}
	Let us compare  $\m (\boldsymbol{\tilde\sigma}^{j-1})$ and $\m (\boldsymbol{\tilde\sigma}^{j})$: In $\boldsymbol{\tilde\sigma}^{j} $ we delete $\disp_{\boldsymbol{\tilde\sigma}^{j-1}}(f(a_j)+1)$ symbols that were not deleted in $\boldsymbol{\tilde\sigma}^{j-1}$. However, on each such a deletion in $\boldsymbol{\tilde\sigma}^{j-1}$ we payed for a mismatch. Next, in both $\boldsymbol{\tilde\sigma}^{j-1}, \boldsymbol{\tilde\sigma}^{j}$ we pay no edit operations as long as we match between $\bxi$ characters and then: 
	
	If $\disp_{\boldsymbol{\tilde\sigma}^{j-1}}(s(a_{j+1}))> 0$: In $\boldsymbol{\tilde\sigma}^{j}$ we pay $\disp_{\boldsymbol{\tilde\sigma}^{j-1}}(s(a_{j+1}))$ additions the prefix of $b_{j+1}$. On the other hand, in 
	$\boldsymbol{\tilde\sigma}^{j-1}$ we pay this much of substitutions caused by matching $\bxi$ characters into $b_{j+1}$.
	
	If $\disp(\boldsymbol{\tilde\sigma}^{j-1}(s(a_{j+1}))< 0$: In this case in $\boldsymbol{\tilde\sigma}^{j-1}$ we pay at least $\disp(\boldsymbol{\tilde\sigma}^{j-1}(s(a_{j+1}))$ mismatches and characters insertions till we reach an index $i$ $\disp(\boldsymbol{\tilde\sigma}^{j-1}(i)\ge s(a_{j+1})$. On the other hand in $\boldsymbol{\tilde\sigma}^{j}$ we pay this much of characters deletions to delete the prefix of $a_{j+1}$. Overall, the total costs of the alignments are equal.
	
	\paragraph{Case 2 -- $\disp(\boldsymbol{\tilde\sigma}^{j-1}(f(a_j)+1))< 0$:} We handle similarly using the same arguments.
%	In this case in $\boldsymbol{\tilde\sigma}^{j}$ we pay for at most $\disp(\boldsymbol{\tilde\sigma}^{j-1}(f(a_j)+1))$ deletions 
%	First observe that: $\disp(\boldsymbol{\tilde\sigma}^{j-1}f(a_j)+\ell) > \disp(\boldsymbol{\tilde\sigma}^{j-1}(f(a_j)+1))$ ().
%	
%	
%	Observe that it cannot be the case that: $\boldsymbol{\tilde\sigma}^{j-1}(f(a_j)+1) > f(a_j)+1$ while $\boldsymbol{\tilde\sigma}^{j-1}(f(a_j)+\ell) < f(a_j)+\ell$, otherwise $\boldsymbol{\tilde\sigma}^{j-1}$ is non-optimal. So assume: $\boldsymbol{\tilde\sigma}^{j-1}(f(a_j)+1)>f(a_j)+1$ and 

To prove the moreover part, observe that $\boldsymbol{\tilde\sigma}'$ matches all the $\bxi$ blocks to themselves. Therefore, if there exists a block $j\in [k]$ which $\boldsymbol{\tilde\sigma}'$ does not align optimally $a_j$ into $b_j$ then there exists a better alignment than $\boldsymbol{\tilde\sigma}'$. Contradicting to the optimality of $\boldsymbol{\tilde\sigma}'$.
	\end{proof}
	
	To conclude the correctness of our algorithm $\Dec$, observe that it mimics the behavior $\boldsymbol{\tilde\sigma}'$ on the $j^{\text{th}}$ block of 
	$(\mathbf{a,b})$. Since $\Delta (\boldsymbol{\tilde\sigma}')_j=\Delta_e(a_j,b_j)$ the proof follows.
	Finally, we note that the bound on the functions $S$ and $T$ hold as $d'=(nk+nk^2)2|\Xi|\le 3|\Xi|dk^2 $.
\end{proof}

\ED problem has also been of special interest in the last few years thanks to the advancements in fine grained complexity \cite{BI18,AHWW16,Williams15,Williams16,Vir18}. The above direct product feasibility immediately implies the below hardness amplification result.

\begin{corollary}\label{cor:Edit}
Let $\varepsilon>0$. Let $\Sigma$ be a finite non-empty set. Let $\Xi$ be a superset of $\Sigma$ of cardinality $|\Sigma|+1$. 
Let  $\mathcal{D}=\{D_d\}_{d\in\mathbb{N}}$ be a family of distributions such that for every randomized algorithm $\mathcal{A}$ running in time $d^{1+\varepsilon}$ over inputs of size $d$, the following holds for large $d\in\mathbb{N}$.
\begin{itemize}
\item $D_d$ is a distribution over $I_{\EDsig}(d)$ where an instance of $I_{\EDsig}(d)$ can be sampled from $D_d$ in $\tilde{O}(d)$ time.
\item $
\underset{(a,b)\sim D_d}{\Pr}\left[\mathcal A \text{ finds optimal \ED alignment for }(a,b) \text{ with probability at least }2/3\right]\le 1-d^{-o(1)}$.
\end{itemize}

Then there is some $\varepsilon'>0$ such that there is a distribution family $\mathcal{D}'=\{D_d'\}_{d\in\mathbb{N}}$ such that for every randomized algorithm $\mathcal{A}'$ running in time $d^{1+\varepsilon'}$, the following holds for all large enough $d\in\mathbb{N}$.
\begin{itemize}
\item $D_d'$ is a distribution over $I_{\EDxi}(d^{1+o(1)})$ and an instance in $I_{\EDxi}$ can be sampled from $D_d'$ in $d^{1+o(1)}$ time.
\item $
\underset{(a',b')\sim D_d'}{\Pr}\left[\mathcal A' \text{ finds optimal \ED alignment for }(a',b') \text{ with probability at least }2/3\right]\le 0.01.
$
\end{itemize}
\end{corollary}
\begin{proof}
We apply Theorem~\ref{thm:main} by setting, $\Pi=I_{\EDsig}$, $\Lambda=I_{\EDxi}$, $p=\nicefrac{2}{3}$, $S(d,k)=3|\Xi|dk^2$, $T(d,k)=2\cdot S(d,k)$, $v(d)=w\cdot d$  (for some $w\in\mathbb{N}$), $s(d)=\tilde{O}(d)$, $t(d)=d^{1+\varepsilon}$, and $\p(d)=\nicefrac{1}{d^{o(1)}}$. Then we have that $k:=d^{o(1)}$ and $c:=300\ln 3$. We verify that $k\cdot s(d)+T(d,k)+v(d)\le \frac{t(d)}{2c}$ holds by noting that $k\cdot s(d)+T(d,k)+v(d)=d^{1+o(1)}$ and $\frac{t(d)}{2c}=\Omega(d^{1+\varepsilon})$. Therefore, we have that the sampling time from $D'_d$ is $d^{1+o(1)}$ and that the theorem statement holds for any  randomized algorithm $\mathcal{A}'$ running in time $ \frac{t(d)}{2c}=\Theta(d^{1+\varepsilon})=\Theta\left(\left(d^{1+o(1)}\right)^{1+\varepsilon-\delta}\right)$, for any $\delta>0$. 
\end{proof}

Strangely, our technique fails to immediately extend to another similarity search problem, namely computing the Fr\'echet distance, which is studied together with edit distance and LCS in literature. 

\begin{remark}[Fr\'echet Distance]\label{rem:frechet}
Note that another well studied subquadratic hard \emph{\cite{B14,AHWW16}} problem of computing Fr\'echet Distance does not seem to admit direct product feasibility (at least the natural way to aggregate fails).
\end{remark}
\subsection{Matrix Multiplication}

Now we consider the matrix multiplication problem. It may be written as an optimization problem in the following rather mundane way: 

\begin{definition}[Matrix Multiplication problem]
Let $q$ be some large prime (universal constant). The Matrix Multiplication ($\Mult$) is an optimization problem charaterized by the following quadruple of objects $(I_{\Mult},\SOL_{\Mult},\m_{\Mult},\min)$, where:
\begin{itemize}
\item For every $d\in\mathbb{N}$, $I_{\Mult}(d)$ is the set of all pairs of $n\times n$ matrices with entries in $\mathbb{F}_q$;
\item For every $(A,B)\in I_{\Mult}$ we have $\SOL_{\Mult}(A,B)$ is the set of all $n\times n$ matrices with entries in $\mathbb{F}_q$;
\item For every $(A,B)\in I_{\Mult}$ and every $n\times n$ matrix $C$ we define $\m_{\Mult}(A,B,C)$ to be the number of entries in $C$ which differ from $AB$.
\end{itemize}
\end{definition}

Given the above formalism, we show that it is self direct product feasible. 

\begin{lemma}\label{lem:MatMult}
Let $S,T:\mathbb{N}\times\mathbb{N}\to\mathbb{N}$, where $S(d,k)=(dk)^2$ and $T(d,k)=O(d^2k^2)$. Then we have that 
$\Mult$ is $(S,T)$-self direct product feasible.
\end{lemma}
\begin{proof}

We define the pair of deterministic algorithms $(\G,\Dec)$ below.	

For every $(A_1,B_1), \dots , (A_k,B_k) \in I_{\Mult}(d)$  given as input to $\G$, it outputs the instance $(A',B')$ in $I_{\Mult}(d')$ defined as:
$$
A' := 
\left[
\begin{array}{c c c}
A_1 & 0 & 0\\
0 & \ddots & 0 \\
0 & 0 & A_k
\end{array}
\right]
$$
and:
$$
B' := 
\left[
\begin{array}{c c c}
B_1 & 0 & 0\\
0 & \ddots & 0 \\
0 & 0 & B_k
\end{array}
\right]
$$

It is clear that the running time of $\G$ is $O(d')$. 

Next, for every $i\in[k]$, $(A_1,B_1),\ldots ,(A_k,B_k)\in I_{\Mult}(d)$, and a matrix $C'\in \SOL_{\Mult}(A',B')$ given as input to $\Dec$, the algorithm first runs $\G$ to compute $(A',B')$ and then outputs $C\in \Mult(A_i,B_i)$ which is the $i$-th block of $C'$.

We next show that if $C'$ is an optimal $\Mult$ matrix for $(A',B')$ then $C$ is an optimal alignment for $(A_i,B_i)$.
An optimal solution for  $(A',B')$ is the matrix $C'=A'\cdot B'$. By  matrix multiplication definition its $i$-th block of $C'$ equals $A_i\cdot B_i$ which is indeed the optimal solution for $(A_i,B_i)$.

Clearly the running time of $\Dec$ is $O(d^2k^2)$, the proof follows. 
\end{proof}

It is a well-known open problem  \cite{K18, WW18}  to find an efficient way to check if the product of two matrices is equal to the third matrix. However, it can be checked in linear time (in the input size) if we allow randomness \cite{F77}. But Theorem~\ref{thm:main} needs the verification to be deterministic so we cannot invoke it directly. 
Let us briefly explain how to modify the proof so it generalizes to the case of matrix multiplication.

First, observe that using~\cite{F77} result, given $n\times n$ matrices $A,B$, and $C$ in $\mathbb{F}_q$ one can check whether  $C=A\cdot B$ with failure probability less than $2^{-m}$ in time $m n^2$. 

Recall the proof of Theorem~\ref{thm:main} and  let us restate it in matrix multiplication terminology: 
Given input matrices $A,B$, the algorithm repeatedly define matrices $A',B'$ such that $A,B$ are nested as a random sub-block of $A',B'$ as defined in the proof of Lemma~\ref{lem:MatMult}. This above process is repeated for $c$-times. The matrices $A',B'$ are fed into algorithm $\mathcal{A}'$, and the output of $\mathcal{A}'$ is the product $A'\cdot B'$.
%
%It defines an algorithm $\mathcal{A}$, such that if there is a randomized algorithm $\mathcal{A}'$ with success probability $p$  running in time $\nicefrac{t(d)}{2c}$ that finds an optimal feasible solution of an instance sampled from $D_d'$ with probability (over the sampling) greater than $0.01$ then, $\mathcal{A}$ has a success probability $p$ running in time $t(d)$ and finds an optimal feasible solution for an instance sampled from $D_d$ with probability (over the sampling) greater than $1-\p(d)$.

We modify the algorithm so that it iterates over its main loop for $2c $ iterations. In such a way, using the same argument described in the proof of Theorem~\ref{thm:main}, we are guaranteed that, on at least $1-\p(d)$ fraction of instances sampled from original distribution, the probability that none of the matrices output by $\mathcal{A'}$ equals $A'\cdot B'$ is at most:  $$\left(1-\frac{\ln\left(\nicefrac{1}{1-p}\right)}{c}\right)^{2c}\le 1- (e^{\ln{(1-p)}})^2\le 1-2p. $$

Now suppose that  on at least one of the iterations $\mathcal{A'}$ outputs a matrix $C'$ that equals $A' \cdot B'$ (which now happen with probability at least $2p$). We run~\cite{F77} algorithm on each matrix output by $\mathcal{A'}$, so that we are guaranteed that with probability at least $1-p$ for each such a matrix, the algorithm outputs 'equal' iff $C'=A'\cdot B'$. In total, the success probability of our modified algorithm is at least $2p-p=p$.

Summarizing, on at least $1-\p(d)$ fraction of inputs sampled from the original distribution, our modified algorithm outputs $A\cdot B$ with probability at least $p$, as claimed. Thus, with a setting of parameters as in the proofs of Corollaries~\ref{cor:LCS}~and~\ref{cor:Edit}, we have the following. 

\begin{corollary}\label{cor:matmult}
Let $\varepsilon>0$. 
Let  $\mathcal{D}=\{D_d\}_{d\in\mathbb{N}}$ be a family of distributions such that for every randomized algorithm $\mathcal{A}$ running in time $d^{1+\varepsilon}$ over inputs of size $d$, the following holds for large $d\in\mathbb{N}$.
\begin{itemize}
\item $D_d$ is a distribution over $I_{\Mult}(d)$ where an instance of $I_{\Mult}(d)$ can be sampled from $D_d$ in $\tilde{O}(d)$ time.
\item $
\underset{(A,B)\sim D_d}{\Pr}\left[\mathcal A \text{ finds correct product of }(A,B) \text{ with probability at least }2/3\right]\le 1-d^{-o(1)}$.
\end{itemize}

Then there is some $\varepsilon'>0$ such that there is a distribution family $\mathcal{D}'=\{D_d'\}_{d\in\mathbb{N}}$ such that for every randomized algorithm $\mathcal{A}'$ running in time $d^{1+\varepsilon'}$, the following holds for all large enough $d\in\mathbb{N}$.
\begin{itemize}
\item $D_d'$ is a distribution over $I_{\Mult}(d^{1+o(1)})$ and an instance in $I_{\Mult}$ can be sampled from $D_d'$ in $d^{1+o(1)}$ time.
\item $
\underset{(A',B')\sim D_d'}{\Pr}\left[\mathcal A' \text{ finds correct product of }(A',B') \text{ with probability at least }2/3\right]\le 0.01.
$
\end{itemize}
\end{corollary}

\section{Almost Worst Case to Average Case for Problems in \TFNP}\label{sec:TFNP}

In this section, we look at two  problems in \TFNP: Factoring and End of a Line. We also point towards future directions of research in the intersection of hardness amplification and \TFNP. 

\subsection{Factoring}

Factoring is the problem of given a number as input, the task of determining all its prime factors. Given a candidate set of prime factors for a number it is possible to efficiently  check that they are prime numbers, factors of the input number, and exhaustive, i.e., there are no more prime factors of the input number. This is captured in our formalism for optimization problems as follows.

\begin{definition}[Prime factor set of a number]
For every positive integer $n>1$ and every set $\Gamma\subseteq [n]$, we say that $\Gamma$ is a prime factor set of $n$ if the following holds.
\begin{itemize}
\item \textbf{Divisor}: For every $a\in \Gamma$, we have that $a$ divides $n$.
\item \textbf{Prime}:  For every $a\in \Gamma$,  we have that $a$ is prime (greater than 1).
\end{itemize}
\end{definition}

\begin{definition}[Factoring problem]
The Factoring problem ($\Fac$) is an optimization problem charaterized by the following quadruple of objects $(I_{\Fac},\SOL_{\Fac},\m_{\Fac},\max)$, where:
\begin{itemize}
\item For every $d\in\mathbb{N}$, $I_{\Fac}(d)=\{2,\ldots ,2^{d}+1\}$; 
\item For every $n\in I_{\Fac}$ we have $\SOL_{\Fac}(n)$ is the set of all prime factor sets of $n$;
\item For every $n\in I_{\Fac}$ and every prime factor set $\Gamma$ of $n$, we define $\m_{\Fac}(n,\Gamma)$ to be the cardinality of $\Gamma$.
\end{itemize}
\end{definition}

Now, we show that factoring is self direct product feasible.

\begin{lemma}\label{lem:factor}
 Let $S,T:\mathbb{N}\times\mathbb{N}\to\mathbb{N}$, where $S(d,k)=dk$ and $T(d,k)=\tilde O(dk)$. Then we have that 
$\Fac$ is $(S,T)$-self direct product feasible.
\end{lemma}
\begin{proof}
We define the pair of deterministic algorithms $(\G,\Dec)$ below.

For every $n_1,\ldots ,n_k\in I_{\Fac}(d)$ given as input to $\G$, it outputs the instance $n'$ in $I_{\Fac}(d')$ (where $d'\le dk$) defined as:
$$
n' := \underset{i\in k}{\Pi} n_i.
$$

It is clear that the running time of $\G$ is $\tilde O(d')$. 

Next, for every $i\in[k]$, $n_1,\ldots ,n_k\in I_{\Fac}(d)$, and a $\Fac$ assignment $\Gamma'\in \SOL_{\Fac}(n)$ for $n'$ given as input to $\Dec$, the algorithm first runs $\G$ to compute $n'$ and then outputs $\Gamma\in \SOL_{\Fac}(n)$ which is computed as follows: We initialize $\Gamma=\emptyset$. Then we scan $\Gamma'$ and for each $a'\in \Gamma'$ we insert $a'$ in $\Gamma$ if $a'$ divides $n_i$.

We next show that if $\Gamma'$ is an optimal $\Fac$ factor set for $n'$ then $\Gamma$ is an optimal factor set for $n_i$. 
Indeed if $\Gamma'$ is an optimal solution for $n'$ then it equals the set of \textit{all} prime factors of $n'$, which include all prime factors of $n_i$. 
Hence each prime factor $a$ of $n$  will be inserted into $\Gamma$ during the scan of $\Gamma'$.

Clearly the running time of $\G$ is $\tilde O(dk)$, since the check whether each $a\in \Gamma'$ divides $n_i$ takes $\tilde O(d)$-time, the proof follows.
\end{proof}

%The factoring problem is known to be in the class \PPA\ (under randomized reductions) \cite{J16}. It's also the basis of many cryptosystems. Therefore its hardness amplification  is of interest. 

The factoring problem is known to be in the class \PPA\ (under randomized reductions) \cite{J16}. It's also the basis of many cryptosystems. Therefore the following hardness amplification result is of interest. 

\begin{corollary}\label{cor:factor}
Let $a\ge 8$.
Let  $\mathcal{D}=\{D_d\}_{d\in\mathbb{N}}$ be a family of distributions such that for every randomized algorithm $\mathcal{A}$ running in time $O(d^a)$ over inputs of size $d$, the following holds for large $d\in\mathbb{N}$.
\begin{itemize}
\item $D_d$ is a distribution over $I_{\Fac}(d)$ where an instance of $I_{\Fac}(d)$ can be sampled from $D_d$ in $\tilde O(d)$ time.
\item $
\underset{n\sim D_d}{\Pr}\left[\mathcal A \text{ finds largest prime factor set of }n \text{ with probability at least }2/3\right]\le 1-\frac{1}{d}$.
\end{itemize}

Then for $m:=\Theta(d^7)$,  there is a distribution family $\mathcal{D}'=\{D_d'\}_{d\in\mathbb{N}}$ such that for every randomized algorithm $\mathcal{A}'$ running in time $O(m^{a/7})$ over inputs of size $m$, the following holds for all large enough $d\in\mathbb{N}$.
\begin{itemize}
\item $D_d'$ is a distribution over $I_{\Fac}(m)$ and an instance in $I_{\Fac}$ can be sampled from $D_d'$ in $\tilde{O}(m)$ time.
\item $
\underset{n'\sim D_d'}{\Pr}\left[\mathcal A' \text{ finds largest prime factor set of }n' \text{ with probability at least }2/3\right]\le 0.01.
$
\end{itemize}
\end{corollary}
\begin{proof}
We apply Theorem~\ref{thm:main} by setting, $\Pi=\Lambda=\Fac$ , $p=\nicefrac{2}{3}$, $S(d,k)=dk$, $T(d,k)=(dk)^{1+o(1)}$, $v(d)=d^{1+o(1)}$, $s(d)=\tilde{O}(d)$, $t(d)=d^a$, and $\p(d)=\nicefrac{1}{d}$. Then we have that $k:=\Theta(d^6)$ and $c:=300\ln 3$. We verify that $k\cdot s(d)+T(d,k)+v(d)\le \frac{t(d)}{2c}$ holds by noting that $k\cdot s(d)+T(d,k)+v(d)=\tilde{O}(d^7)$ and $\frac{t(d)}{2c}=\Omega(d^a)=\Omega(d^8)$ (because $a\ge 8$). Therefore, we have that the sampling time from $D'_d$ is $\tilde{O}(d^7)=\tilde{O}(m)$ and that the theorem statement holds for any  randomized algorithm $\mathcal{A}'$ running in time $ \frac{t(d)}{2c}=\Theta(d^a)=\Theta(m^{a/7})$. 
\end{proof}

We wonder if the above result can have any meaningful, concrete applications in cryptography.

\subsection{End of a Line and Nash Equilibrium}\label{sec:EOL}

The End of a Line problem \cite{P94,DGP09} is the canonical \PPAD-complete problem. Informally it captures the handshaking lemma in directed graphs on $2^n$ vertices given as input through predecessor and successor circuits of $\poly(n)$ size. It may be written as an optimization problem under our formalism in the following (mundane) way.

\begin{definition}[Size of a Circuit]
For every positive integer $n$ and every circuit $C:\{0,1\}^n\to\{0,1\}^n$, the size of $C$ denoted by $|C|$ is the number of bits needed to describe $C$.
\end{definition}

\begin{definition}[End of a Line problem]
The End of a Line problem ($\EOL$) is an optimization problem charaterized by the following quadruple of objects $(I_{\EOL},\SOL_{\EOL},\m_{\EOL},\max)$, where:
\begin{itemize}
\item For every $d\in\mathbb{N}$, $I_{\EOL}(d)=\{(A,B)\mid A,B:\{0,1\}^n\to\{0,1\}^n\text{ and }|A|+|B|\le d\}$; 
\item For every $(A,B)\in I_{\EOL}$ we have $\SOL_{\EOL}(A,B)=\{0,1\}^n$;
\item For every $(A,B)\in I_{\EOL}$ and every $x\in\{0,1\}^n$, we define $\m_{\EOL}(A,B,x)$ as follows:
$$
\m_{\EOL}(A,B,x)=\begin{cases}
1\text{ if }A(B(x))\neq x\text{ or }B(A(x))\neq x\neq 0^n\\
0\text{ otherwise}
\end{cases}
$$
\end{itemize}
\end{definition}

\begin{lemma}\label{lem:EOL}
 Let $S,T:\mathbb{N}\times\mathbb{N}\to\mathbb{N}$, where $S(d,k)=dk$ and $T(d,k)=O(dk)$. Then we have that 
$\Fac$ is $(S,T)$-self direct product feasible.
\end{lemma}
\begin{proof}
We define the pair of deterministic algorithms $(\G,\Dec)$ below.

Fix $k,d\in \mathbb N$.
For every $(A_1,B_1),\ldots ,(A_k,B_k)\in I_{\EOL}(d)$ given as input to $\G$, it outputs the instance $(A',B')$ in $I_{\EOL}(d')$ where $d'=dk$, and $A',B':\{0,1\}^{nk}\to\{0,1\}^{nk}$ are defined as follows. For all $x'=(x_1',\ldots ,x_k')\in\{0,1\}^{nk}$, where $\forall i\in[k]$, we have $x_i'\in\{0,1\}^n$, and define $A'(x')$ and $B'(x')$ as follows.
\begin{align*}
A'(x')&=(A_1(x_1'),\ldots ,A_k(x_k'))\\
B'(x')&=(B_1(x_1'),\ldots ,B_k(x_k'))
\end{align*}

It is clear that the running time of $\G$ is $O(d')$. 

Next, for every $i^*\in[k]$, $(A_1,B_1),\ldots ,(A_k,B_k)\in I_{\EOL}(d)$ , and a $\EOL$\ solution $x'\in \zo^{nk}$ for $(A',B')$ given as input to $\Dec$, the algorithm first runs $\G$ to compute $(A',B')$ and then outputs $x\in \zo^n$ which is computed as follows:
$x_j=x'_{(i^*-1)n+j}$.

We next show that if $x'$ is an optimal $\EOL$\ solution for $(A',B')$ then $x$ is an optimal solution for $(A_{i^*},B_{i^*})$. 
This follows easily by the fact that the $(A_i,B_i)$s are defined on disjoint union of circuits, hence an optimal solution for $(A',B')$ induces an optimal solution for each of the $(A_i,B_i)$s. 

The running times  of $\G$ and $\Dec$ trivially follow.
\end{proof}

\begin{corollary}\label{cor:EOLpoly}
Let $a\ge 8$.
Let  $\mathcal{D}=\{D_d\}_{d\in\mathbb{N}}$ be a family of distributions such that for every randomized algorithm $\mathcal{A}$ running in time $O(d^a)$ over inputs of size $d$, the following holds for large $d\in\mathbb{N}$.
\begin{itemize}
\item $D_d$ is a distribution over $I_{\EOL}(d)$ where an instance of $I_{\EOL}(d)$ can be sampled from $D_d$ in $\tilde O(d)$ time.
\item $
\underset{(A,B)\sim D_d}{\Pr}\left[\mathcal A \text{ finds a solution of }(A,B) \text{ with probability at least }2/3\right]\le 1-\frac{1}{d}$.
\end{itemize}

Then for $m:=\Theta(d^7)$,  there is a distribution family $\mathcal{D}'=\{D_d'\}_{d\in\mathbb{N}}$ such that for every randomized algorithm $\mathcal{A}'$ running in time $O(m^{a/7})$ over inputs of size $m$, the following holds for all large enough $d\in\mathbb{N}$.
\begin{itemize}
\item $D_d'$ is a distribution over $I_{\EOL}(m)$ and an instance in $I_{\EOL}$ can be sampled from $D_d'$ in $\tilde{O}(m)$ time.
\item $
\underset{(A',B')\sim D_d'}{\Pr}\left[\mathcal A' \text{ finds a solution of }(A',B') \text{ with probability at least }2/3\right]\le 0.01.
$
\end{itemize}
\end{corollary}
\begin{proof}
We apply Theorem~\ref{thm:main} by setting, $\Pi=\Lambda= \EOL$, $p=\nicefrac{2}{3}$, $S(d,k)=dk$, $T(d,k)=w\cdot dk$ (for some $w\in\mathbb{N}$), $v(d)=w'\cdot d$  (for some $w'\in\mathbb{N}$), $s(d)=\tilde{O}(d)$, $t(d)=d^a$, and $\p(d)=\nicefrac{1}{d}$. Then we have that $k:=\Theta(d^6)$ and $c:=300\ln 3$. We verify that $k\cdot s(d)+T(d,k)+v(d)\le \frac{t(d)}{2c}$ holds by noting that $k\cdot s(d)+T(d,k)+v(d)=\tilde{O}(d^7)$ and $\frac{t(d)}{2c}=\Omega(d^a)=\Omega(d^8)$ (because $a\ge 8$). Therefore, we have that the sampling time from $D'_d$ is $\tilde{O}(d^7)=\tilde{O}(m)$ and that the theorem statement holds for any  randomized algorithm $\mathcal{A}'$ running in time $ \frac{t(d)}{2c}=\Theta(d^a)=\Theta(m^{a/7})$. 
\end{proof}

Since there is a direct many-one reduction from $\EOL$ to the problem of computing approximate Nash equilibrium in various kinds of games \cite{CDT09,R18,R16}, the above corollary extends to give hardness amplification results for computing Nash equilibrium as well.

\begin{remark}
The proof of Lemma~\ref{lem:EOL} can be mimicked to get self direct product feasibility of canonical complete problems of other \TFNP\ classes like 
Leaf (complete for \PPA) \cite{P94},
LocalOPT (complete for \PLS) \cite{JPY88,DP11},
End of Metered Line (equivalent to \textsf{EoPL})\cite{HY17},
and
Unique EOPL \cite{FGMS18}.
\end{remark}

Now we consider the same hardness amplification result of $\EOL$ but against subexponential time algorithms. We provide below a slightly informal statement (using asymptotic notations as opposed to providing specific constants) for clarity.

\begin{corollary}\label{cor:EOLETH}
Let  $\mathcal{D}=\{D_d\}_{d\in\mathbb{N}}$ be a family of distributions such that for every randomized algorithm $\mathcal{A}$ running in time $2^{o(d)}$ over inputs of size $d$, the following holds for large $d\in\mathbb{N}$.
\begin{itemize}
\item $D_d$ is a distribution over $I_{\EOL}(d)$ where an instance of $I_{\EOL}(d)$ can be sampled from $D_d$ in $\tilde O(d)$ time.
\item $
\underset{(A,B)\sim D_d}{\Pr}\left[\mathcal A \text{ finds a solution of }(A,B) \text{ with probability at least }2/3\right]\le 1-\frac{1}{2^{o(d)}}$.
\end{itemize}

Then for $m:=2^{o(d)}$,  there is a distribution family $\mathcal{D}'=\{D_d'\}_{d\in\mathbb{N}}$ such that for every randomized algorithm $\mathcal{A}'$ running in time $m^{\omega(1)}$ over inputs of size $m$, the following holds for all large enough $d\in\mathbb{N}$.
\begin{itemize}
\item $D_d'$ is a distribution over $I_{\EOL}(m)$ and an instance in $I_{\EOL}$ can be sampled from $D_d'$ in $\tilde{O}(m)$ time.
\item $
\underset{(A',B')\sim D_d'}{\Pr}\left[\mathcal A' \text{ finds a solution of }(A',B') \text{ with probability at least }2/3\right]\le 0.01.
$
\end{itemize}
\end{corollary}
\begin{proof}Let $h:\mathbb{N}\to\mathbb{N}$ be some slowly increasing function such that $\underset{x\rightarrow \infty}{\lim}\ h(x)=\infty$ and let $h:=h(d)$. 
We apply Theorem~\ref{thm:main} by setting, $\Pi=\Lambda=\EOL$, $p=\nicefrac{2}{3}$, $S(d,k)=dk$, $T(d,k)=w\cdot dk$ (for some $w\in\mathbb{N}$), $v(d)=w'\cdot d$  (for some $w'\in\mathbb{N}$), $s(d)=\tilde{O}(d)$, $t(d)=2^{d/h}$, and $\p(d)=\nicefrac{1}{2^{d/(6\cdot h^2)}}$. Then we have that $k:=\Theta(2^{d/h^2})$ and $c:=300\ln 3$. We verify that $k\cdot s(d)+T(d,k)+v(d)\le \frac{t(d)}{2c}$ holds by noting that $k\cdot s(d)+T(d,k)+v(d)=2^{(1+o(1))\cdot d/h^2}$ and $\frac{t(d)}{2c}=\Omega(2^{d/h})$. Therefore, we have that the sampling time from $D'_d$ is $2^{(1+o(1))\cdot d/h^2}=O(m)$ and that the theorem statement holds for any  randomized algorithm $\mathcal{A}'$ running in time $ \frac{t(d)}{2c}=O(2^{d/h})=\Theta(m^{h(d)})=m^{\omega(1)}$. 
\end{proof}

Note that if we assume the (randomized) Exponential Time Hypothesis (\ETH{}) for \PPAD\ \cite{BPR16} then we obtain a family of distributions $\mathcal{D}=\{D_d\}_{d\in\mathbb{N}}$  such that  for every randomized algorithm $\mathcal{A}$ running in time $2^{d^{1-o(1)}}$ over inputs of size $d$, the following holds for large $d\in\mathbb{N}$.
\begin{itemize}
\item $D_d$ is a distribution over $I_{\EOL}(d)$ where an instance of $I_{\EOL}(d)$ can be sampled from $D_d$ in $\tilde O(d)$ time.
\item $
\underset{\phi\sim D_d}{\Pr}\left[\mathcal A \text{ finds maximizing assignment for }\phi \text{ with probability at least }2/3\right]\le 1-\frac{1}{2^{d^2}}$.
\end{itemize} 

Therefore, our amplification in Corollary~\ref{cor:EOLETH} is an \emph{almost} worst-case to average-case reduction for $\EOL$ (under subexponential time reductions).

\iftrue
\subsubsection*{Acknowledgements}

We would like to thank Amir Abboud, Irit Dinur, and Eylon Yogev for discussions and comments.
\fi

\bibliographystyle{alpha}
\bibliography{references}

\appendix

\section{Missing Proofs}\label{sec:missing}
An elementary proof of Lemma~\ref{lem:DP} may be found in the lecture notes of Guruswami and O'Donnell \cite{O05} and we reproduce it below.

\begin{proof}[Proof of Lemma~\ref{lem:DP}]
We first show that the following holds:
\begin{align}
\underset{\substack{x\sim \D\\i\sim [k]}}{\E}[(\mu_{i,x}-\mu)^2]\le \frac{1}{k}.\label{eqap1}
\end{align}

In order to prove \eqref{eqap1}, we  introduce some notations: 
\begin{itemize}
	\item We denote elements of $X^k$ by $\bar x$, and denote by $\D^k$ the $k$-wise product distribution of $\D$.
	\item Fix $i\in [k], x\in X$, let $\D^{k,-i,x}$ be the distribution over $\bar x \in X^k$ obtained by picking $x_1,\dots,x_{i-1},x_{i+1},\dots, x_k$ from $\D$ (independently) and setting $\bar x = (x_1,\dots, x_{i-1},x,x_{i+1},\dots, x_k)$.
	\item Fix $i\in [k]$, define $F^{i}:X\to \R$ and $\sigma_i\in\mathbb{R}$ as follows: 
	 $$\forall x\in X,\ F^i(x)\eqdef \underset{\bar x \sim \D^{k,-i,x}}{\E}[f(\bar x)-\mu]\text{ and }\sigma_i\eqdef \underset{x\sim\D}{\E}[F^i(x)^2].$$
	\item Fix $i\in [k]$ define $f^{i}:X^k\to \R$ by setting for all $\bar x\in X^k$, $f^{i}(\bar x)\eqdef F^{i}(\bar x_i)$.
\end{itemize}

With these notations setup, we have:
$$
	\underset{\substack{x\sim \D\\i\sim [k]}}{\E}[(\mu_{i,x}-\mu)^2]= \underset{\substack{i\sim [k]}}{\E}\left[\underset{\substack{x\sim \D}}{\E} \left[F^{i}(x)^2\right]\right] = \frac{1}{k}\sum_{i\in [k]}\sigma_i.
$$

Therefore, in order to show \eqref{eqap1}, it suffices to show:

\begin{align}
\sum_{i\in [k]}\sigma_i \le 1.\label{eqap2}
\end{align}

We observe that from linearity of expectation, the following holds:
\begin{align}
\forall i\in[k],\ \underset{\substack{x\sim \D}}{\E} [F^i(x)]=0.\label{eqap3}
\end{align}

Also from the independent choice of coordinates in $\bar x$ and \eqref{eqap3}, we have for all $i, j\in [k]$ such that $i\neq j$ we have: 
\begin{align}
\underset{\bar x\sim \D^k }{\E}\left[f^i(\bar x)f^j(\bar x)\right]=0.\label{eqap4}
\end{align}

Finally we define $g:X^k \to \R$ as follows:
\[
\forall \bar x\in X^k,\ g(\bar x)=\sum_{i\in [k]}f^i(\bar x).
\]

Notice that
\begin{eqnarray}\label{eq:g}
0\le \underset{\bar x\sim \D^k }{\E}\left[(f(\bar x)-g(\bar x))^2\right]=\underset{\bar x\sim \D^k }{\E}\left[f^2(\bar x)\right]-2\cdot \underset{\bar x\sim \D^k }{\E}\left[f(\bar x)g(\bar x)\right]+\underset{\bar x\sim \D^k }{\E}\left[g^2(\bar x)\right].	
\end{eqnarray}

We bound each of the above terms separately. 
First,  we have that $\underset{\bar x\sim \D^k }{\E}\left[f^2(\bar x)\right]=\mu$ since $f$ is a boolean valued function.
Second, we have
\begin{align*}
\underset{\bar x\sim \D^k }{\E}\left[f(\bar x)g(\bar x)\right]&=\underset{\bar x\sim \D^k }{\E}\left[f(\bar x)\cdot \sum_{i\in [k]}f^i(\bar x)\right]&\\
&=\sum_{i\in [k]}\underset{\bar x\sim \D^k }{\E}\left[f(\bar x)\cdot f^i(\bar x)\right]&\\
&= \sum_{i\in [k]}\underset{x\sim \D }{\E} \left[F^i(x) \cdot\underset{\bar x \sim \D^{k,-i,x}}{\E}\left[f(\bar x) \right]\right]&\\
&= \sum_{i\in [k]}\underset{x\sim \D }{\E} \left[F^i(x) \cdot(\mu +F^i(x))\right]&\left(\text{because }F^i(x)= \underset{\bar x \sim \D^{k,-i,x}}{\E}[f(\bar x)-\mu]\right)\\
&= \sum_{i\in [k]}\mu\cdot \underset{x\sim \D }{\E} \left[F^i(x) \right]+ \sum_{i\in [k]}\underset{x\sim \D }{\E} \left[\left(F^i(x) \right)^2\right]&\\
&=  \sum_{i\in [k]}\underset{x\sim \D }{\E} \left[\left(F^i(x) \right)^2\right]&(\text{from }\eqref{eqap3})\\
&=\sum_{i\in [k]}\sigma_i.
\end{align*}

Third, we have
\begin{align*}
\underset{\bar x\sim \D^k }{\E}\left[g^2(\bar x)\right]&=\underset{\bar x\sim \D^k }{\E}\left[\left(\sum_{i\in [k]}f^i(\bar x)\right)^2\right]&\\
&=\underset{\bar x\sim \D^k }{\E}\left[\sum_{i,j\in [k]}f^i(\bar x)\cdot f^j(\bar x)\right]&\\
&=\underset{\bar x\sim \D^k }{\E}\left[\sum_{\substack{i,j\in [k]\\ i\neq j}}f^i(\bar x)\cdot f^j(\bar x)\right]+\underset{\bar x\sim \D^k }{\E}\left[\sum_{i\in [k]}\left(f^i(\bar x)\right)^2\right]&\\
&=\underset{\bar x\sim \D^k }{\E}\left[\sum_{i\in [k]}\left(f^i(\bar x)\right)^2\right]&(\text{from }\eqref{eqap4})\\
&=\sum_{i\in [k]}\sigma_i.&
\end{align*}

Plugging all together in \eqref{eq:g}, we get:
\[
0 \le \mu -2\cdot \sum_{i\in [k]}\sigma_i + \sum_{i\in [k]}\sigma_i.
\]
Since $\mu\le 1$ and we get $\sum_{i\in [k]}\sigma_i\le 1$, as claimed. Thus we have shown \eqref{eqap2} holds and consequently \eqref{eqap1} holds. The proof of the lemma then follows from a simple application of Markov inequality: 
\begin{align*}
\Pr_{\substack{x\sim \D\\i\sim [k]}}\left[|\mu_{i,x}-\mu|\ge \frac{1}{k^{1/3}}\right]&=\Pr_{\substack{x\sim \D\\i\sim [k]}}\left[(\mu_{i,x}-\mu)^2\ge \frac{1}{k^{2/3}}\right]&\text{}\\
&\le k^{2/3}\cdot \underset{\substack{x\sim \D\\i\sim [k]}}{\E}[(\mu_{i,x}-\mu)^2]&\text{(from Markov inequality)}\\
&\le \frac{k^{2/3}}{k}&\text{(from \eqref{eqap1})}\\
&=\frac{1}{k^{1/3}}&\hfill\qedhere
\end{align*}

\end{proof}

\end{document}